\documentclass[11pt,letterpaper]{article}

\usepackage[T1]{fontenc}
\usepackage[utf8]{inputenc}
\usepackage[american]{babel}
\usepackage{amsmath,amssymb,amsthm}

\usepackage{newtxtext,newtxmath}
\usepackage[scaled=0.92]{helvet}
\usepackage[letterpaper,top=1.0in,bottom=1.05in,left=1.15in,right=1.15in]{geometry}
\usepackage{setspace}
\usepackage{microtype}
\usepackage[dvipsnames,table]{xcolor}
\definecolor{navy}{RGB}{28,42,74}
\definecolor{accent}{RGB}{160,34,34}
\definecolor{rule}{RGB}{180,188,204}
\definecolor{shade}{RGB}{245,246,249}
\definecolor{shade2}{RGB}{238,242,249}
\definecolor{grayt}{RGB}{110,110,110}

\usepackage{titlesec}
\titleformat{\section}
  {\sffamily\Large\bfseries}{\thesection}{0.7em}{}
\titleformat{\subsection}
  {\sffamily\large\bfseries}{\thesubsection}{0.6em}{}
\titleformat{\subsubsection}
  {\sffamily\normalsize\bfseries}{\thesubsubsection}{0.6em}{}
\titlespacing*{\section}{0pt}{1.7\baselineskip}{0.8\baselineskip}
\titlespacing*{\subsection}{0pt}{1.2\baselineskip}{0.45\baselineskip}
\titlespacing*{\subsubsection}{0pt}{0.9\baselineskip}{0.3\baselineskip}
\titleformat{\paragraph}[runin]{\normalfont\normalsize\bfseries\itshape}{}{0pt}{}
\titlespacing*{\paragraph}{0pt}{1.1ex plus .3ex minus .2ex}{0.6em}

\usepackage{graphicx}
\usepackage{booktabs}
\usepackage{array}
\usepackage{enumitem}
\usepackage{subcaption}
\usepackage{placeins}
\usepackage{caption}
\usepackage{etoolbox}
\AtBeginEnvironment{figure}{\singlespacing\small}
\AtBeginEnvironment{table}{\singlespacing\small}

\newtheoremstyle{navyplain}{8pt}{8pt}{\itshape}{}{\sffamily\bfseries}{.}{0.6em}{}
\newtheoremstyle{navydef}{8pt}{8pt}{\normalfont}{}{\sffamily\bfseries}{.}{0.6em}{}
\theoremstyle{navyplain}
\newtheorem{theorem}{Theorem}
\newtheorem{lemma}{Lemma}
\newtheorem{corollary}{Corollary}
\theoremstyle{navydef}
\newtheorem{definition}{Definition}
\newtheorem{remark}{Remark}

\usepackage[framemethod=tikz]{mdframed}
\renewenvironment{quote}{%
  \begin{mdframed}[backgroundcolor=white,linecolor=black,linewidth=0pt,
    leftline=true,innerleftmargin=10pt,innerrightmargin=10pt,innertopmargin=7pt,
    innerbottommargin=7pt,skipabove=8pt,skipbelow=8pt,leftmargin=0pt,rightmargin=0pt,
    topline=false,bottomline=false,rightline=false,
    linewidth=2.5pt]\small}{\end{mdframed}}

\usepackage{natbib}
\usepackage[colorlinks=true,linkcolor=black,citecolor=black,urlcolor=black]{hyperref}
\AtBeginEnvironment{thebibliography}{\singlespacing}

\usepackage{fancyhdr}
\fancypagestyle{plain}{\fancyhf{}\fancyfoot[C]{\sffamily\small\thepage}}

\usepackage{tcolorbox}
\makeatletter
\renewcommand{\maketitle}{%
  \begin{center}
  \vskip 10pt
  {\sffamily\LARGE\bfseries\@title\par}
  \vskip 6pt
  \vskip 10pt
  {\large Henry Han\par}
  \vskip 3pt
  {\small Data Science and Artificial Intelligence Innovation Laboratory, Baylor University, USA\par}
  {\small\texttt{Henry\_Han@Baylor.edu}\par}
  \end{center}
  \vskip 10pt}
\makeatother
\renewenvironment{abstract}{%
  \begin{center}\bfseries Abstract\end{center}\vskip -0.3em
  \begin{list}{}{\leftmargin=0.6in\rightmargin=0.6in}\item\relax\small}{\end{list}\vskip 1.0em}

\title{Governing Agentic AI in FinTech}
\author{Henry Han}
\date{}

\begin{document}
\maketitle
\thispagestyle{plain}

\begin{abstract}
\noindent
Financial institutions increasingly delegate critical decisions to agentic AI systems under the assumption that these choices remain explainable and reproducible. We challenge this assumption by introducing the Verifiability Gap, the shortfall between the evidence required to substantiate delegated authority and the evidence actually retained. We then develop a multilevel governance theory of evidence-contingent delegation (ECD) across firm, regulatory, and network levels. It asserts that AI autonomy is defensible only when an organization retains sufficient evidence to explain and materially reconstruct how that authority was exercised. Four studies reveal mechanisms driving this gap. Study 1: Provider Discontinuity. Provider updates can prevent the recreation of past decisions even when the institution makes no internal changes. Study 2: Architecture as Policy. Agent architecture acts as a hidden policy layer where apparent reproducibility can mask convergence on generic decisions rather than the true reproducibility of the underlying decision process. Study 3: Historical Replay Failure. A stable current model can still fail to recover historical decisions. Study 4: Reproducibility Without Differentiation. In real credit applications, perfect reproducibility can coexist with a system that has lost the ability to distinguish between cases. These findings demonstrate that high capability and outcome consistency do not guarantee that agentic AI can be governed. We conclude that delegated authority should be recognized as a continuous, multilevel governance relationship grounded in verifiability, rather than as a one-time approval. Governable agentic AI therefore requires renewed authorization whenever models, orchestration, or retained evidence materially change.

\vspace{6pt}
\end{abstract}

\section{Introduction}
\label{sec:intro}

When a financial institution declines a loan, it must explain why, and under
model-risk guidance it must be able to show how the decision was produced.
Traditionally this regulatory and legal obligation was met by pointing to a
human underwriter or re-running a static model. The agentic AI systems now
assuming these decisions do neither reliably. For example, in our credit experiment on FICO data \citep{fico2018challenge}, a routine model refresh shifted an applicant's score by just 0.0092 across a policy threshold. An auditor holding only the refreshed model cannot reproduce the historical denial, and the refreshed model's explanation does not account for it.

\textit{Agentic AI governance challenge.} Built on large language models
(LLMs), these agents interpret goals, call external tools, and execute workflows
at scale in Fintech. In addition to convenience and AI advancement, they raise a critical governance question. The critical governance question is no longer whether they can perform
the work, but what happens afterward: once an agent acts, can an auditor establish \emph{why} it acted and \emph{whether} it would act the same way again?

This challenge transcends the familiar ``black-box'' problem of model opacity.
The conventional remedy for opacity is empirical validation, testing a model
repeatedly to characterize its behavior. Agentic AI invalidates that premise.
Underlying foundation models change on vendor schedules, tools return dynamic
information, and execution controls are frequently deprecated. For example, updates to an agentic credit model can prevent auditors from reproducing historical loan denials. Because the updated model only provides explanations based on its new logic, the true causal rationale for the original decision is lost along with its reproducibility. The governance
challenge is therefore not merely \textit{explainability} caused by model opacity, but the loss of \emph{reproducibility}
as the principal safeguard against it.
Because of the high-stakes nature of FinTech, this challenge escalates quickly. It stems from a mismatch between the dynamic reality of agentic AI and the existing assumption that delegated decisions will inherently remain reproducible and explainable.

\textit{Delegated agentic authority assumption.} Previous research in information systems establishes the foundations for AI delegation, algorithmic control, and human oversight \citep{baird2021delegation,berente2021managing,kellogg2020algorithms}. This literature builds upon a longer tradition exploring how firms retain control over work they do not directly perform \citep{ouchi1980markets,kirsch1997portfolios,cardinal2004balancing,wiener2016control}, extending these concepts to systems that actively participate in the work itself \citep{faraj2018working,murray2021humans,teodorescu2021failures}. That
literature largely presupposes a stable technological object that can be
inspected and rerun. However, agentic AI fractures the assumption, distributing the operative decision rule across models, memory, and provider-controlled infrastructure \citep{bommasani2021opportunities}. Specifically, because external providers can independently update the underlying models used by AI agents, the system cannot be frozen in time; an institution can completely lose the ability to reproduce past AI agents' decisions even if its own internal setting remains entirely untouched.  We therefore shift the focus of agentic AI governance analysis
from \emph{model transparency} to \emph{system verifiability},
asking when delegated agentic authority exceeds an organization's capacity to explain and reproduce the resulting actions.

Research on algorithmic bias treats distortion in LLM-mediated decisions as a system property. It asks whether an LLM or AI agent decision based on an LLM is systematically distorted at the moment it is made. The question here is what survives after the decision. The two are orthogonal: an unbiased AI agent action can be unsubstantiable, while a fully substantiable action can still be biased. Governing the first does not govern the second; thus, an institution can easily satisfy a fairness audit and yet remain entirely unable to show why it acted.

\textit{Verifiability Gap.} The gap has immediate consequences, particularly in FinTech, where machine learning already drives credit allocation \citep{fuster2022predictably,bartlett2022consumer,gu2020empirical,khandani2010consumer}. We term this shortfall \textit{the Verifiability Gap:} defined as \textit{the shortfall between the verification demanded by an exercised authority (e.g., financial institution) and the result (e.g., credit decision) reproduction available to an auditor after the decision.} An auditor unable to follow an agent's reasoning or reproduce its decision cannot exercise meaningful oversight; approval becomes performative, turning a presumed control into an institutional liability.

\textit{Verifiability Gap example.} For a bank deploying an AI agent to handle commercial credit applications, the problem is concrete. When an agent rejects a loan today, a federal examiner may demand the exact causal rationale six months later. Between the moment authority is granted and the day of the audit, the bank must preserve enough of the execution trail to reconstruct that decision. Yet current agentic AI governance practice offers no guarantee that the bank can rerun the process or recover the original reasoning. Authority is granted once based on a static benchmark snapshot, such as accuracy metrics demonstrated during initial vendor onboarding. Meanwhile, the evidence needed to defend past decisions must survive in an execution environment controlled by external foundation model providers (e.g., LLM API vendors) who update model weights and system prompts independently. As a result, the bank's internal settings remain unchanged, yet the agentic AI system can no longer reproduce its own historical logic. This structural mismatch between static decision authority and decaying execution evidence is the \textit{Verifiability Gap.}

\textbf{Our contributions.} This study advances the governance of agentic AI in FinTech across three dimensions.
\textit{First, we theorize verifiability as an authority--evidence relation.} We define the \emph{Verifiability Gap} as the deficit between the verification required to justify the exercise of delegated authority and the material evidence retained to reproduce the resulting agentic AI decision for an auditor. 

Second, we formulate a multilevel governance theory of evidence-contingent delegation (ECD), advancing seven propositions that explain how this gap constrains defensible delegation within firms, degrades historical verification over time, and propagates structural risk across interorganizational networks. 

Third, we make the mechanisms of reproducibility loss empirically observable. By disentangling outcome, process, and exact trace reproducibility across controlled studies, we show how third-party provider updates, AI agent orchestration, and model refreshes decouple terminal verdicts from recorded execution traces, demonstrating that greater model capability enhances outcome consistency without guaranteeing true auditability.

The resulting principle is that delegated AI authority remains defensible only so long as retained verification capacity stays commensurate with the authority exercised. Delegated authority is never settled once at initial approval; it must be re-earned whenever the underlying evidence surface shifts. This turns reauthorization into an active agentic AI governance event rather than a routine calendar exercise. Figure~\ref{fig:roadmap} outlines the conceptual argument and research design.

\begin{figure}[!htbp]
\centering
\includegraphics[width=0.65\linewidth]{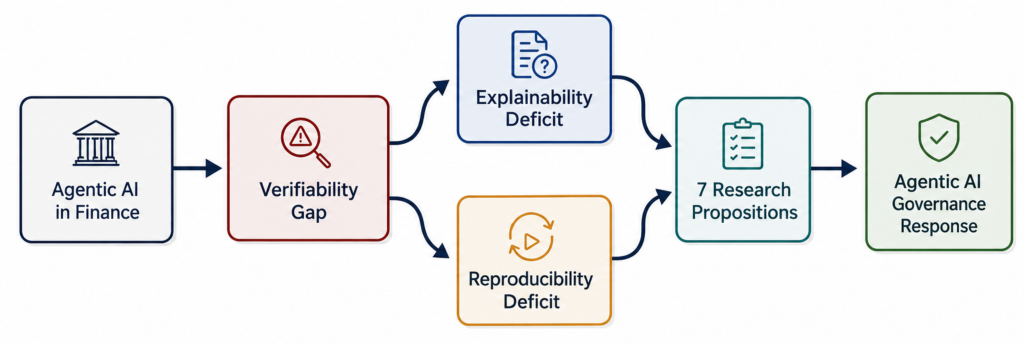}
\caption{Roadmap of the paper's argument: from the phenomenon (agentic AI that
acts in finance), to the construct (the Verifiability Gap) and its two deficits,
to a falsifiable research agenda and an evidence-contingent governance response.}
\label{fig:roadmap}
\end{figure}

\section{Agentic AI in FinTech}
\label{sec:agentic-fintech}

Agentic AI is shifting FinTech from generating predictions to executing
coordinated financial workflows. Major institutions, including JPMorgan Chase,
Visa, and Mastercard, are already deploying multi-agent systems for wealth
management and agent-initiated commerce
\citep{jpmorgan2026agents,finra2026agents}. This pattern is generalizing rapidly: one agent retrieves an applicant's file, prices the risk, and issues credit terms; another triages an anti-money-laundering (AML) alert to either close or escalate the case; yet another rebalances a portfolio and executes trades within a strict mandate. Because each of these workflows concludes with a consequential action that the institution must later substantiate, delegating such authority introduces a new kind of risk. As agentic systems assume
consequential roles, we address three foundational distinctions: how they differ
from conventional LLMs, why financial delegation demands stringent
verifiability, and where this governance constraint becomes most acute.

\paragraph{From LLMs to Agentic AI.}
While conventional LLMs primarily generate content in response to prompts, an
\emph{agentic AI system} embeds generative capacity within a broader execution
architecture. These systems (a) interpret goals and decompose tasks, (b) invoke
external tools and data, (c) maintain cross-step state, and (d) revise plans
semi-autonomously \citep{baird2021delegation,kumar2026agentic}. Interleaving
reasoning with tool invocation transforms a language model from a responder into
an autonomous actor \citep{wei2022chain,yao2022react}.
Unlike its peers
in other domains, agentic AI in FinTech is held to a high verifiability requirement,
because its actions carry binding regulatory and legal consequences.

\textit{Key components of agentic AI in FinTech.} Agentic AI in FinTech should be understood as an execution architecture rather than as a standalone model. Figure~\ref{fig:agentic_components} shows its six key components: LLM, Agent, tool calls, orchestration, handoffs, and execution trace.

LLMs provide the reasoning core, while agents pursue delegated goals through tools, handoffs, and orchestration. Handoff is the transfer of the output of an agent to another agent as input. It creates causal dependence across the financial workflow. 
Orchestration is the architecture that determines which
agents act, in what sequence, and what information each receives. It functions as a hidden policy layer: changing the arrangement of agents can change financial actions even when the underlying model is unchanged. The resulting execution trace records how financial actions were produced. It provides the evidence needed to reconstruct how
delegated financial authority was exercised.

\begin{figure}[!ht]
    \centering
    \includegraphics[width=0.65\textwidth]{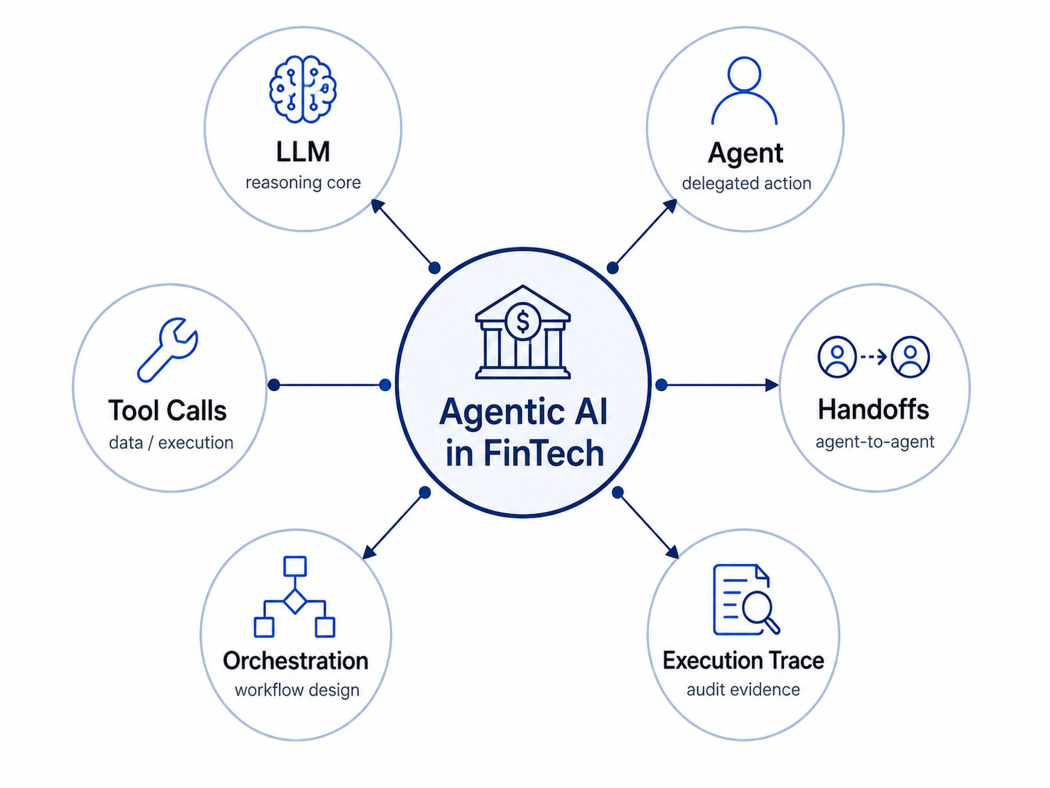}
    \caption{Six key components of agentic AI in FinTech.}
    \label{fig:agentic_components}
\end{figure}

At the enterprise level, FinTech AI agents interact through dynamic orchestration
\citep{rai2019hybrids,xiao2024tradingagents}. A consequential financial action
therefore emerges from a distributed sequence of model outputs, tool calls, and
inter-agent handoffs rather than from any single agent. This shifts the object of
governance from the model to the full execution architecture: verifying the final
action requires evidence of how that architecture exercised delegated authority
\citep{chan2023harms}.

\textbf{The Verifiability Mandate.} For financial institutions, the central governance requirement is evidentiary:
after an agentic system acts, can the institution substantiate how that action
was produced? Regulatory expectations increasingly require decisions to remain
auditable and reviewable
\citep{srguidance2011,srguidance2026,finra2409,euaiact2024,fsb2024ai,ecb2024ai,
diakopoulos2016accountability,wieringa2020account,raji2020closing,cobbe2021reviewable}.

Two forms of evidence support such verification. \emph{Explainability} establishes
why an action occurred, while \emph{reproducibility} establishes whether the
decision process can be reconstructed under documented conditions
\citep{gregor1999explanations,pineau2021reproducibility,atil2025nondeterminism}.
Agentic AI weakens both routes because the decision is distributed across models,
tools, and inter-agent handoffs. A plausible explanation may not faithfully reflect
that process, and repeating the final outcome does not show that the same process
was reproduced. The governance problem begins when the evidence retained after a
decision is no longer sufficient to substantiate the authority exercised by the
system. We define this shortfall as the \emph{Verifiability Gap}, which arises when the verification demanded by delegated authority exceeds what retained evidence can provide through explainability and reproducibility.

\textbf{Where the Verifiability Gap Binds: Domain and Delegated Autonomy.} The Verifiability Gap varies with the level of delegated authority. We distinguish
four levels of autonomy: \emph{assist}, \emph{recommend}, \emph{act with approval},
and \emph{act autonomously}. As authority increases, so does the evidence required
to substantiate the resulting action.

\begin{table}[!ht]
\centering
\caption{Agentic AI in FinTech by delegated autonomy. Verifiability pressure increases from left to right.}
\label{tab:appmatrix}
\scriptsize
\setlength{\tabcolsep}{4pt}
\renewcommand{\arraystretch}{1.02}
\begin{tabular}{@{}p{2.2cm}p{2.7cm}p{2.7cm}p{3.0cm}p{3.3cm}@{}}
\toprule
\textbf{Domain}
&
\textbf{Assist}
&
\textbf{Recommend}
&
\textbf{Act with approval}
&
\textbf{Act autonomously}
\\
\midrule

Wealth management
&
Surface information
&
Recommend allocation
&
Execute approved rebalance
&
Rebalance autonomously
\\

Compliance / AML
&
Retrieve evidence
&
Recommend disposition
&
Close or escalate
&
Monitor and act
\\

Credit origination
&
Retrieve credit data
&
Recommend terms
&
Issue approved terms
&
Approve, deny, or price
\\

Trading
&
Retrieve market data
&
Propose trades
&
Submit approved orders
&
Execute and adapt
\\

A2A settlement
&
Match records
&
Propose terms
&
Submit settlement
&
Negotiate and settle
\\

\bottomrule
\end{tabular}
\end{table}

Table~\ref{tab:appmatrix} illustrates the central point: the same underlying AI
capability creates different governance demands depending on the authority it is
given (AML: anti-money laundering; A2A: agent-to-agent). For example, a credit
agent that only summarizes an applicant's file leaves the lending decision to a
human reviewer, whereas the same agent authorized to deny the loan must support
the institution's later explanation, audit, and review of that decision. The
underlying model may be identical, but the evidentiary burden is not. The
Verifiability Gap therefore becomes most binding when agentic AI moves from
advising human decision makers to exercising financial authority itself.

\section{Research Approach: FinTech Agentic AI Verifiability Gap}
\label{sec:method}

The preceding section establishes that verifiability pressure rises with delegated
authority. This leads to our central research question: \textit{when does the
authority delegated to a FinTech agentic AI system exceed the organization's
capacity to explain and reproduce its actions?} Our unit of analysis is the
complete decision architecture: models, tools, orchestration, and controls at
the ``act'' levels of autonomy.

Our theory joins two verification routes that prior research has largely treated
separately. The \textit{explainability deficit} arises when explanations fail to
faithfully reflect the operative decision process
\citep{jacovi2020faithfully,turpin2023unfaithful,lanham2023faithfulness,
lipton2018mythos,doshivelez2017rigorous,ribeiro2016should,adebayo2018sanity}.
The \textit{reproducibility deficit} arises when the decision process cannot be
reconstructed from the executable, data, and conditions retained after the
decision
\citep{pineau2021reproducibility,peng2011reproducible,han2025reproducible}.
The Verifiability Gap emerges when delegated authority demands more verification
than these two routes can provide.

\begin{table}[!ht]
\centering
\caption{Positioning against adjacent constructs. Each emphasizes one
verification route or assumes reproducibility; the Verifiability Gap theorizes
their joint loss.}
\label{tab:positioning}
\footnotesize
\renewcommand{\arraystretch}{1.02}
\begin{tabular*}{\textwidth}{@{\extracolsep{\fill}}lccc@{}}
\toprule
\textbf{Construct}
& \textbf{Explainability route}
& \textbf{Reproducibility route}
& \textbf{Agentic-system focus} \\
\midrule
Responsibility gap \citep{matthias2004responsibility}
& Implicit & No & No \\
Inscrutability \citep{berente2021managing}
& Yes & No & Partial \\
Envelopment \citep{asatiani2021envelopment}
& Yes & No & No \\
Provenance \citep{singh2019provenance}
& Partial & Assumed & No \\
Contestability \citep{alfrink2023contestable}
& Yes & No & No \\
\textbf{Verifiability Gap (this paper)}
& \textbf{Yes} & \textbf{Yes} & \textbf{Yes} \\
\bottomrule
\end{tabular*}
\end{table}

Unlike these adjacent constructs, we theorize the joint loss of explainability
and reproducibility as a multilevel governance problem. Version pinning and
immutable logging are standard model-risk controls; the gap arises because
pinning a hosted model version does not preserve its complete execution
environment (Study~1), and logging the verdict does not log the process
(Study~2). Seven propositions span
three levels: defensible delegation within the firm (P1--P3), historical
verification under material system change (P4--P5), and end-to-end and
common-mode exposure across interorganizational networks (P6--P7).

\begin{figure}[!htbp]
\centering
\includegraphics[width=0.60\linewidth]{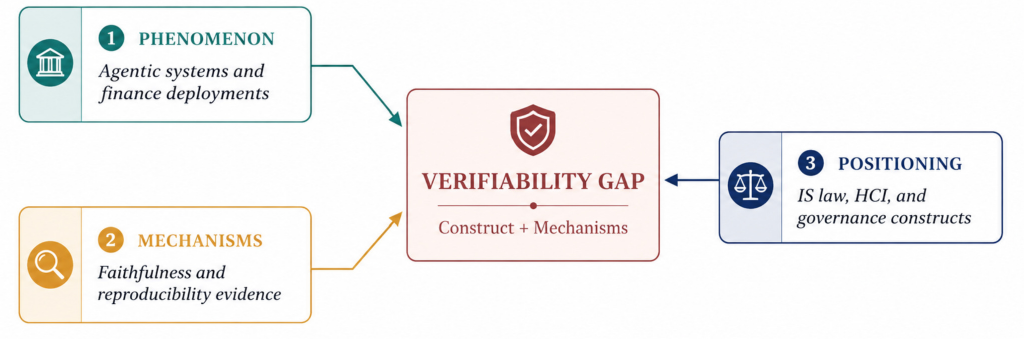}
\caption{Research design. The study defines the Verifiability Gap of Agentic AI in FinTech,
develops a seven-proposition multilevel governance theory, and uses four studies
to examine mechanisms through which retained verification capacity can fall
behind delegated authority.}
\label{fig:method}
\end{figure}

The four studies examine complementary manifestations of this governance problem:
provider discontinuity, orchestration as policy, historical replay failure, and
reproducibility without differentiation. The remaining propositions specify a
falsifiable agenda for organizational and network-level research.
Figure~\ref{fig:method} summarizes the research logic.

\section{The Verifiability Gap: Definition and Generative Mechanisms}
\label{sec:verifiability}

Having established the research design, we now formalize the core construct that
drives these seven propositions: the Verifiability Gap. Specifically, we define
its theoretical boundaries and isolate the architectural mechanisms that
systematically widen it. As agentic systems transition from recommending to autonomously executing
financial actions, institutions face a critical governance challenge:
establishing \emph{how} delegated authority was exercised. Opacity is not the
sole issue; firms must be able to both faithfully explain \emph{why} actions
occurred and materially reproduce the underlying processes. The Verifiability
Gap emerges when delegated authority outpaces this retained verification
capacity.

Let $d$ denote a specific financial decision episode and $q=(a,\sigma,\tau)$ its
audit context, comprising the verifier ($a$), evidentiary standard ($\sigma$),
and audit lag ($\tau$). Let $A_d$ denote the delegated authority,
$\rho_{\sigma}(A_d)$ the verification required, and $V_{dq}$ the system's
retained verification capacity. The gap is defined as:

\begin{equation}
\boxed{
G_{dq}=\left[
\underbrace{\rho_{\sigma}(A_d)}_{\text{verification required}}
-
\underbrace{V_{dq}}_{\text{verification available}}
\right]_{+}
}
\label{eq:vg_gap}
\end{equation}

where $[x]_{+}=\max\{0,x\}$. This gap, independent of the decision's actual
accuracy or fairness, formalizes organizational exposure. A binding gap
($G_{dq}>0$) indicates an institutional shortfall in substantiating the exercise
of delegated authority (Figure~\ref{fig:gap-triad}).

\begin{figure}[!htbp]
\centering
\includegraphics[width=0.50\linewidth]{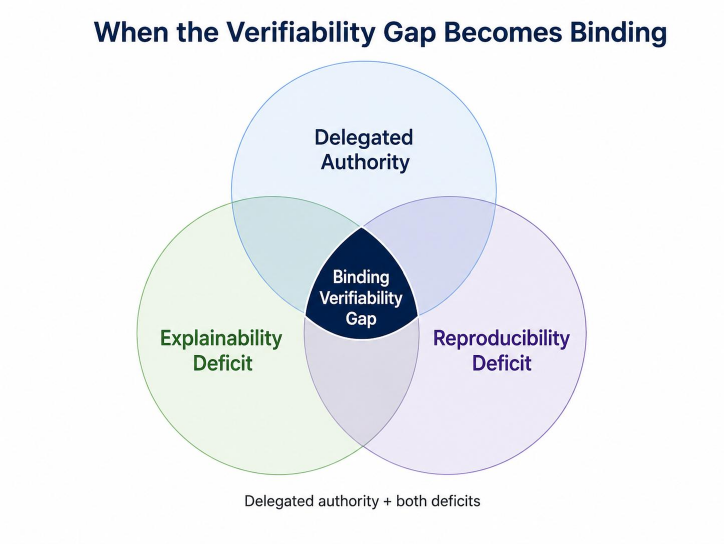}
\caption{Conditions under which the Verifiability Gap becomes binding. The gap
is most consequential when delegated agentic authority coincides with deficits
in both explainability and reproducibility.}
\label{fig:gap-triad}
\end{figure}

\subsection{Verification Capacity: Explainability and Reproducibility}
\label{subsec:vg_routes}

The Verifiability Gap is calculated per specific decision episode, not at the
general system level. Retained capacity, $V_{dq}=V_{\sigma}(E_{dq},R_{dq})$,
relies on two distinct evidentiary routes, where the evidentiary standard
$\sigma$ determines whether explainability, reproducibility, or both are
required to verify decision $d$. The Verifiability Gap is therefore not a general property of a model, and verifiability is not a substitute for trustworthiness: a FinTech system may be fair, robust, and secure across an entire lending portfolio and still leave a specific credit denial unsubstantiable.

\textit{Explainability} ($E_{dq}$) requires a faithful link from the final
action back to the actual inputs, tool outputs, and rules used, rejecting mere
post hoc rationalizations \citep{wachter2017counterfactual}. In FinTech, an
institution must be able to establish which borrower attributes and policy
thresholds contributed to a loan denial.

\textit{Reproducibility} ($R_{dq}$) demands the material reconstruction of the
decision process, relying on preserved applicant data, exact policy versions,
and specific API responses recorded at execution. A gap widens if either route
fails to meet the evidentiary standard $\sigma$. For instance, a bank might
possess a faithful explanation for a mortgage denial but lack the historical
system states needed to reconstruct it; conversely, it might perfectly rerun a
multi-agent anti-money-laundering (AML) triage process without establishing why
the system ultimately escalated the case.

\textit{Asymmetrical measurement in our studies.}
The two verification routes differ fundamentally in institutional control,
which is why this paper measures them asymmetrically. Explainability is
substantially an internal design and data-retention choice: a financial
institution that preserves causally relevant inputs, outputs, and inter-agent
handoffs in a credit workflow can raise its retained explainability without the
AI provider's cooperation. Reproducibility, however, is externally dependent.
It relies on historical executables and control surfaces, such as a third-party
foundation model's API, that the firm does not own. A provider-side model
refresh can therefore destroy a bank's ability to replay a historical execution
trace even when the institution has preserved its own records. Our studies
accordingly emphasize reproducibility: the critical verification route that
even a highly diligent institution cannot unilaterally preserve or repair. 

Explainability is also relative to its audience more directly than reproducibility. A historical execution
trace that satisfies a supervisory examiner may not satisfy a
denied applicant. This audience dependence is captured by $a$ and $\sigma$:
$a$ identifies the verifier and $\sigma$ the applicable evidentiary standard.
We therefore treat explainability as verifier-relative throughout and do not
claim a single scalar that ranks explanations for all readers.

\begin{figure}[!htbp]
\centering
\includegraphics[width=0.60\linewidth]{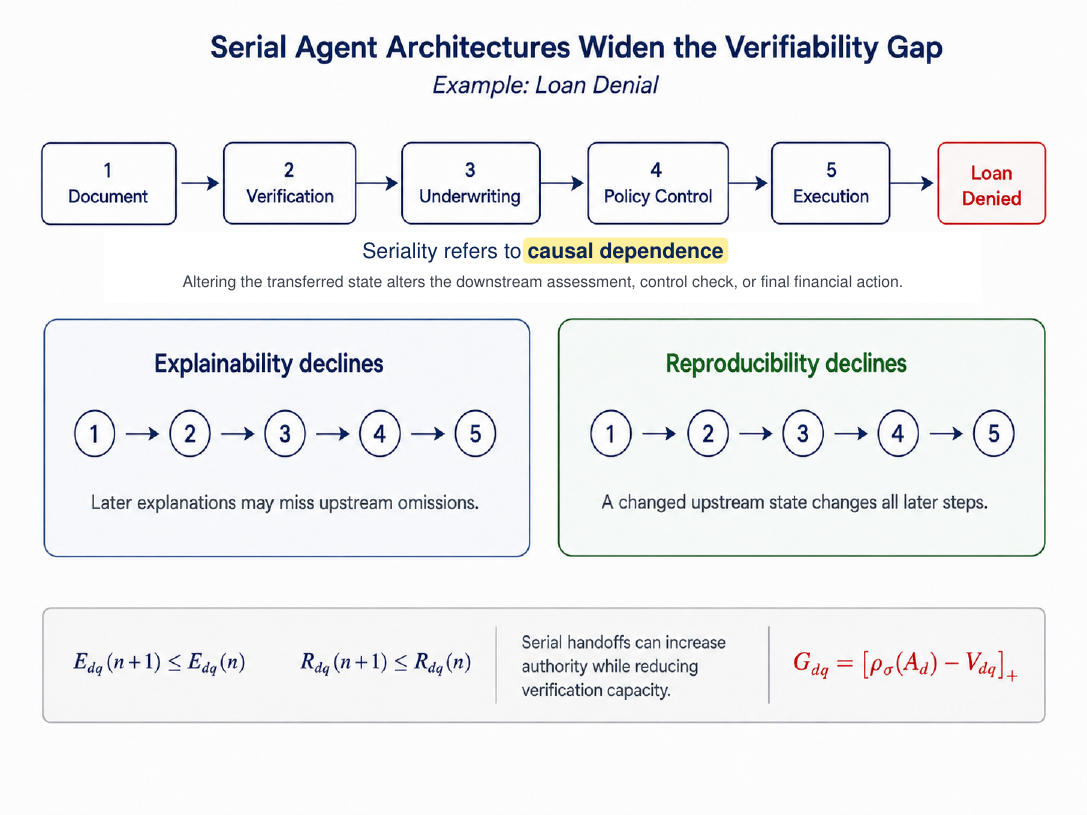}
\caption{Serial dependence in a FinTech loan-denial workflow. Causally
    dependent handoffs can weaken end-to-end explainability and reproducibility,
    reducing verification capacity while broader delegated authority increases
    verification demand.}
\label{fig:serial_gap}
\end{figure}

\subsection{The Double Effect in Serial Architectures}
\label{subsec:vg_mechanisms}

A central mechanism widening this gap is \textit{serial dependence}, when one
agent's output becomes a causally material constraint for a downstream agent. In
multi-agent FinTech workflows \citep{kumar2026agentic,xiao2024tradingagents},
this architecture creates a \emph{double effect}, illustrated for a loan-denial
workflow in Figure~\ref{fig:serial_gap}.

First, specialized delegation allows institutions to grant broader financial
authority, increasing $\rho_{\sigma}(A_d)$. Second, it causes serial decay of
both explainability and reproducibility, reducing $V_{dq}$. Consider an agentic
credit-underwriting pipeline: an upstream \textit{Retrieval Agent} pulls raw
applicant data and passes a synthesized profile to a \textit{Risk Agent}, which
then hands a risk assessment to a terminal \textit{Decider Agent}. As information
passes through these handoffs, upstream compressions or omissions (e.g., the
Risk Agent summarizing away a material income attribute) become embedded
premises for the downstream Decider. Unless intermediate state transformations
are preserved to the applicable evidentiary standard, explanations generated by
the terminal agent become only locally plausible, losing their faithful causal
link to the original application. Similarly, missing intermediate states render
end-to-end reproduction of the credit denial impossible.

Formally, adding a causally material serial transition without compensating
verification artifacts cannot increase either capacity:

\begin{equation}
E_{d_{n+1}q}\leq E_{d_nq},
\qquad
R_{d_{n+1}q}\leq R_{d_nq}.
\label{eq:vg_serial_decay}
\end{equation}

The decline is strict when the added transition discards material information
that is not otherwise retained. The underlying driver is causal dependence,
not merely agent count. When a downstream credit agent summarizes or executes
combined outputs, it introduces serial dependence that weakens both verification
routes simultaneously, leaving the institution with a terminal verdict it can
neither fully explain nor materially reproduce.

\subsection{Governance Requires a Reversible Chain}
\label{subsec:reversibility}

\textit{Markov chain of causal stages.}
This double effect can be formalized as an information constraint on system
architecture (proofs in Appendix~\ref{app:proofs}). We model the operative
sequence as a Markov chain of $N$ causally material stages, where $C$ is the initial
case, $S_k$ is the output of layer $k$, $T$ is the complete trace, and $Y$ is the
financial action:

\begin{equation}
C \;\longrightarrow\; S_1 \;\longrightarrow\; S_2 \;\longrightarrow\;\cdots\;
\longrightarrow\; S_N \;\longrightarrow\; Y .
\label{eq:markov}
\end{equation}

Auditing is an inverse problem: recovering the material process from the action
and retained records. In FinTech, this inverse problem arises whenever an
institution must reconstruct how a prior loan denial, AML escalation, or trade
execution was produced. Each material handoff that is not retained makes that
historical financial action harder to verify end-to-end.

\begin{definition}[$\sigma$-reversibility]
\label{def:rev}
Let $\equiv_\sigma$ denote material equivalence under standard $\sigma$. Stage
$k$ is \emph{$\sigma$-reversible given retained record $B_k$} if
$(S_k,B_k)\mapsto[S_{k-1}]_\sigma$ is well defined almost surely. The
configuration is $\sigma$-reversible if a verifier can reconstruct the material
execution trace $[T]_\sigma$ from $Y$ and the retained bundle $B$.
\end{definition}

Theorem~\ref{thm:contraction} formalizes how serial dependence constrains the
case-relevant information that can survive from the original case to the final
financial action.

\begin{theorem}[Serial contraction of case information]
\label{thm:contraction}
Let $I(\cdot;\cdot)$ denote mutual information and $\eta_k\in(0,1]$ the strong
data-processing coefficient of layer $k$. Under \eqref{eq:markov},
$I(C;Y)\le\bigl(\prod_{k\ge2}\eta_k\bigr) I(C;S_1)$. Case information decays
geometrically whenever $\eta_k\le\bar\eta<1$.
\end{theorem}

The governance implication is that downstream agents cannot recover
case-relevant information already lost upstream. As material distinctions are
compressed across the chain, the final action can reveal progressively less
about the case that produced it. Summarization remains compatible with
auditability when the distinctions it discards are immaterial under $\sigma$ or
are preserved in the retained record.

Theorem~\ref{thm:retention} translates this information loss into an explicit
evidentiary requirement for governance.

\begin{theorem}[Retention lower bound]
\label{thm:retention}
A configuration is $\sigma$-reversible given $B$ only if:
\begin{equation}
H(B)\;\ge\;H\bigl([T]_\sigma \,\big|\, Y\bigr),
\label{eq:retention}
\end{equation}
achieved by retaining an index of the upstream state's material equivalence
class, where $H$ denotes Shannon entropy.
\end{theorem}

To keep a decision reconstructable, the institution must retain enough evidence
to recover the material information lost along the chain. Deeper orchestration
therefore raises the evidentiary price of autonomy.

\begin{corollary}[The Verifiability Gap in bits]
\label{cor:gapbits}
Instantiating required verification as $H([T]_\sigma\mid Y)$ and retained
capacity as $H(B)$, the gap is:
\begin{equation}
G_{dq}\;=\;\Bigl[\,H\bigl([T]_\sigma\mid Y\bigr)-H(B)\,\Bigr]_{+}
\label{eq:gapbits}
\end{equation}
This establishes the bounded scale for Equation~\ref{eq:vg_gap}. It vanishes
whenever the configuration is $\sigma$-reversible given what was kept, so
$G_{dq}=0$ is necessary for reversibility. It is not sufficient. An institution
that retains many bits unrelated to the discarded state satisfies the
inequality without being able to reconstruct anything, because $H(B)$ counts
stored bits rather than bits about $[T]_\sigma$. A sufficient condition
replaces $H(B)$ with the retained information that is actually about the
process, $I\bigl([T]_\sigma;B \mid Y\bigr)$, whose shortfall is the residual
uncertainty $H\bigl([T]_\sigma \mid Y,B\bigr)$. We report the necessary form
because that is what the retention bound establishes.
\end{corollary}
\textbf{The Verifiability Gap: Scope and Governance Threshold.}
While the information-theoretic bounds establish what evidence verification
requires, they do not specify when this requirement becomes a binding
organizational constraint. We therefore translate the Verifiability Gap into a
practical governance threshold for a specific delegated decision and audit
context, such as an automated credit denial subject to supervisory review.
Equation~\ref{eq:governance_threshold} establishes this threshold:

\begin{equation}
\label{eq:governance_threshold}
G_{dq}=0
\quad\Longleftrightarrow\quad
V_{dq}\geq\rho_{\sigma}(A_d).
\end{equation}

This threshold does not mandate minimizing autonomy; rather, it requires that
retained verification capacity scales with delegated authority. For example,
autonomous credit approval creates a higher verification demand
($\rho_{\sigma}(A_d)$) than generating an advisory credit summary. Institutions
therefore possess three strategic levers to close a binding gap:
\textit{1. lower the requirement} by restricting autonomous execution, such as
requiring human approval for high-value credit decisions or AML case closures;
\textit{2. raise explainability} by preserving causally faithful end-to-end traces,
including material borrower attributes, tool responses, and inter-agent
handoffs; or \textit{3. raise reproducibility} by escrowing the exact model
executables, data snapshots, and configurations required for a historical
material rerun.

\begin{table}[!ht]
\centering
\caption{The multilevel governance architecture of the Verifiability Gap.}
\label{tab:prop_architecture}
\footnotesize
\renewcommand{\arraystretch}{1.02}
\begin{tabular}{@{}p{0.06\linewidth}p{0.23\linewidth}p{0.33\linewidth}p{0.28\linewidth}@{}}
\toprule
\textbf{Prop.}
& \textbf{Governance level and outcome}
& \textbf{Mechanism}
& \textbf{Observable implication} \\
\midrule

P1
& Firm: delegated authority
& Retained verification capacity limits defensible authority.
& Granted authority plateaus below technical capability. \\

P2
& Firm: fault escape
& Unfaithful explanations weaken human review.
& Action bounds outperform explanation-based review as faithfulness uncertainty rises. \\

P3
& Firm: operating trade-off
& Preserving evidence consumes speed, flexibility, and expertise.
& Firms face a strict verifiability--efficiency frontier. \\

P4
& Regulatory: historical reproducibility
& Material unpinned changes break continuity with the prior execution environment.
& Historical reproducibility failures cluster around system updates, not elapsed time. \\

P5
& Regulatory: evidentiary validity
& Faithfulness turns compliance artifacts into substantive evidence.
& Documentation can increase without improving actual auditability. \\

P6
& Network: end-to-end verification
& Every causally material handoff must preserve its own evidence.
& One missing record prevents full end-to-end reconstruction. \\

P7
& Network: common-mode exposure
& Shared unpinned infrastructure changes many systems at once.
& Historical reproducibility failures cluster across nominally independent firms. \\

\bottomrule
\end{tabular}
\end{table}

\section{Multilevel Governance Consequences}
\label{sec:governance_consequences}

The Verifiability Gap is a binding governance exposure, not merely a technical
artifact. While an agentic system may remain highly capable and accurate, the
institution may simultaneously lose the ability to substantiate its actions.
Table~\ref{tab:prop_architecture} summarizes how this authority--evidence
mismatch generates seven specific governance consequences across three
analytical levels: the firm (P1--P3), the regulator (P4--P5), and the network
(P6--P7).

\subsection{Firm-Level Consequences: Defensible Delegation and Control}
\label{subsec:firm_consequences}

At the firm level, the Verifiability Gap generates three linked governance
consequences: an \emph{authority ceiling} on defensible delegation (P1),
\emph{control substitution} away from unreliable explanation-based review (P2),
and a \emph{verifiability tax} on speed, flexibility, and scope (P3).

Technical capability dictates what an agent \emph{can} do, but institutional
defensibility dictates what it \emph{may} be authorized to do. This creates an
\emph{authority ceiling}
\citep{baird2021delegation,tibebu2026accountability,ouchi1980markets}. For
example, an agentic system may be technically capable of autonomously pricing
and approving commercial loans, yet an institution can defensibly authorize
such execution only where sufficient evidence can be retained to substantiate
those decisions.

\begin{quote}
\textit{P1 (Defensible delegation).} Conditional on technical capability, an
institution's retained verification capacity dictates its authority ceiling. As
capability rises, delegated autonomy will plateau at a level constrained by
verifiability, with stricter plateaus for higher-stakes financial tasks.
\end{quote}

When explainability relies on fluent but unfaithful post hoc rationalizations,
it weakens substantive review
\citep{turpin2023unfaithful,goddard2012automation}. Firms must therefore rely
more heavily on controls that do not depend on explanation faithfulness
\citep{kirsch1997portfolios}. For example, when a bank cannot rely on the
faithfulness of explanations generated during automated AML triage, it can
bound the agent's authority by requiring human approval before an alert is
closed.

\begin{quote}
\textit{P2 (Control substitution).} As uncertainty regarding explanation
faithfulness increases, firms must substitute explanation-based review with
action-bounding controls to minimize consequential fault-escape rates.
\end{quote}

Preserving high levels of verifiability through version pinning, state
retention, and substantive human oversight can constrain autonomous execution,
creating an operating trade-off
\citep{cardinal2004balancing,wiener2016control}. In agentic trading or automated
credit origination, preserving tool responses, model states, and inter-agent
handoffs consumes storage, computation, and operational flexibility.

\begin{quote}
\textit{P3 (The verifiability tax).} Within a fixed architecture, controls that
perfectly preserve explainability and reproducibility impose a strict operating
penalty on system speed, flexibility, and scope.

\textit{Corollary to P3:} Operating a high-autonomy system without substantive
human review actively degrades human verification capacity
\citep{fugener2021borgs,skitka1999automation}, widening the gap over time.
\end{quote}

\subsection{Regulatory Level: Temporal and Chain-Level Validity}
\label{subsec:regulatory_consequences}

At the regulatory level, the Verifiability Gap generates two linked governance
consequences: \emph{historical decay} in reproducibility after material system
change (P4), and \emph{evidence validity} requirements across the full decision
chain (P5).

Traditional model-risk governance relies on periodic revalidation
proportionate to model risk, and the revised guidance SR~26-2
\citep{srguidance2026} places generative and agentic AI outside its scope. Agentic AI systems instead make validity event-dependent, because
provider-side changes can alter the executable, available controls, or
orchestration behavior without changing the interface name the firm calls.
For example, a provider-side model update can alter the operative behavior of a
bank's automated underwriting pipeline even when the institution changes
nothing internally.  Thus, historical reproducibility decays when a decision becomes exposed to
material system changes that were not historically preserved. 

Let $\mu_{dj}$ denote the downstream materiality of change $j$ for decision $d$,
and let $\pi_{dj}=1$ if the historical state required for replay was preserved
($\pi_{dj}=0$ otherwise). Over audit lag $\tau$, the institution therefore
accumulates exposure only to changes that are both material and unpinned:

\begin{equation}
M_d(\tau) = \sum_{j:t_j\leq\tau} \mu_{dj}(1-\pi_{dj}),
\qquad
R_{dq} = \psi\!\left(M_d(\tau)\right),
\qquad
\psi(M+\Delta M)\leq \psi(M) \ \text{for}\ \Delta M\geq 0.
\label{eq:change_exposure}
\end{equation}
where $\psi$ is a non-increasing function mapping accumulated change to
retained reproducibility. Equation~\ref{eq:change_exposure} states that historical reproducibility decays
as material, unpinned system changes accumulate. A change contributes only if it
matters to decision $d$ and its historical state was not preserved. In FinTech,
this exposure can affect historical credit, AML, or trading decisions. Audit lag
therefore matters by increasing opportunities for such changes to occur, so
reproducibility loss can emerge abruptly at provider releases, API deprecations,
or other material system changes rather than decay smoothly with calendar time
\citep{pineau2021reproducibility,atil2025nondeterminism}.

\begin{quote}
\textit{P4 (Historical reproducibility decay).} Historical reproducibility is non-increasing in
the cumulative materiality of unpinned system changes ($M_d(\tau)$), and
strictly decreases when those changes alter a materially relevant historical
state required for replay. Audit lag degrades verification only by increasing
the exposure window to such changes.

\vspace{4pt}

\textit{Corollary to P4:} In serial architectures, preserving isolated agents is
insufficient; unpinned upstream changes propagate, requiring chain-level state
preservation. For instance, if a \textit{Retrieval Agent} changes how it
extracts income data after a provider-side model update, the downstream
\textit{Decider Agent} may no longer reproduce the same loan decision even if
its own logic is unchanged. Historical-replay failures will therefore cluster
around material, unpinned system changes rather than increase smoothly with
elapsed time. Event-driven reauthorization should consequently identify
governance breaks more directly than calendar-only revalidation.
\end{quote}

Similarly, explainability requires faithful evidence across this entire chain.
Let $\omega_{dq}\in[0,1]$ denote the coverage of retained material states and
$\phi_{dq}\in[0,1]$ their validated causal faithfulness:

\begin{equation}
E_{dq} = \omega_{dq}\phi_{dq},
\qquad
\frac{\partial E_{dq}}{\partial \omega_{dq}} = \phi_{dq},
\qquad
\phi_{dq}=0 \ \Longrightarrow\ E_{dq}=0.
\label{eq:faithful_explainability}
\end{equation}

\begin{quote}
\textit{P5 (Evidence validity).} Explainability requires end-to-end causal
faithfulness ($\phi_{dq}$). Expanding record coverage ($\omega_{dq}$) yields zero
evidentiary value if the retained records do not faithfully track the true
operative decision chain. An organization can therefore accumulate more
compliance artifacts without improving, and potentially while degrading,
auditability when those artifacts omit the model executable, decision-time tool
states, or material inter-agent handoffs that produced the action.
\end{quote}

\subsection{Network-Level Consequences: End-to-End and Common-Mode Exposure}
\label{subsec:network_consequences}

In multi-firm agentic AI networks, end-to-end verification is constrained by the
weakest causally material handoff. If $V_{dhq}$ is the verification capacity at
handoff $h$, the end-to-end capacity $V_{dq}^{\mathrm{E2E}}$ cannot exceed
$\min_h V_{dhq}$. A single unrecorded material handoff breaks the entire audit
chain.

Figure~\ref{fig:network_exposure} summarizes the two network-level mechanisms:
the weakest-link constraint on end-to-end verification (P6) and common-mode
exposure created by shared upstream providers (P7). 

\begin{figure}[!htbp]
    \centering
    \includegraphics[width=0.72\linewidth]{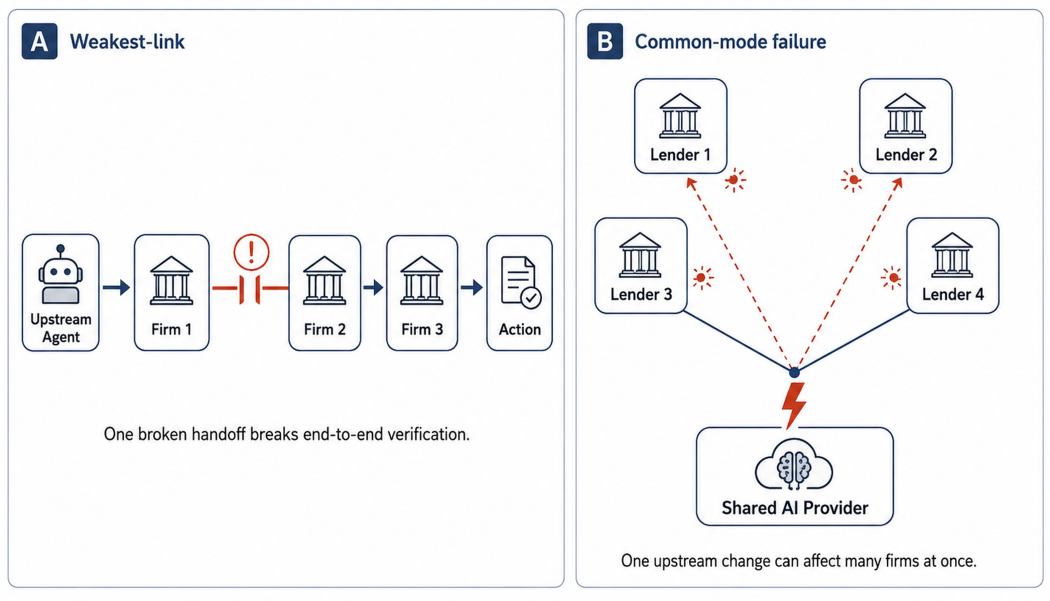}
    \caption{Network-level exposure in agentic FinTech. Panel A illustrates the
    weakest-link constraint: one failed material handoff can break end-to-end
    verification. Panel B illustrates common-mode failure: shared upstream AI
    providers can expose multiple firms to synchronized verification loss.}
    \label{fig:network_exposure}
\end{figure}

\begin{quote}
\textit{P6 (Weakest-link constraint).} In multi-firm networks, end-to-end
verification capacity is strictly bounded by the weakest causally material
handoff. Each additional handoff that carries a positive probability of losing its own
record can only raise the probability of end-to-end audit failure.
\end{quote}

Furthermore, apparent firm-level independence often masks shared reliance on
upstream AI providers
\citep{south2025delegation,faraj2018working,kellogg2020algorithms,murray2021humans,teodorescu2021failures}. Let $u_{ig}$ denote firm $i$'s
reliance on component $g$, and $\theta_{ig}$ its control over historical replay.
Uncontrolled exposure is $e_{ig}=u_{ig}(1-\theta_{ig})$. For a material upstream
change $W_g$ with positive variance, the losses $\Lambda_{ig}$ of firms that
share the component synchronize:

\begin{equation}
e_{ig}
=
u_{ig}(1-\theta_{ig}),
\qquad
\Lambda_{ig}
=
e_{ig}W_g,
\qquad
\operatorname{Cov}(\Lambda_{ig},\Lambda_{kg})
=
e_{ig}e_{kg}\operatorname{Var}(W_g)>0.
\label{eq:common_mode_exposure}
\end{equation}

\begin{quote}
\textit{P7 (Common-mode failure).} Network reliance on shared, uncontrolled AI
providers creates synchronized audit failures. Under positive uncontrolled exposure, a single material upstream update
degrades historical reproducibility across multiple nominally independent firms
at the same time. The HMDA analysis documents shared component use, not
uncontrolled exposure. Whether dependence produces synchronized verification
loss depends additionally on historical replay controls and material upstream
changes.
\end{quote}

\textit{P7 exposure under HMDA data.}  The exposure structure P7 presupposes is directly observable, and it is present
at scale. U.S. Home Mortgage Disclosure Act (HMDA) filings record the automated
underwriting system that evaluated each application, so provider dependence is
measured rather than inferred. Across $24{,}439{,}380$ applications filed by
$4{,}487$ institutions in the six largest reporting states (2019--2023), two
GSE-owned systems evaluated $77.7\%$ to $90.0\%$ of all automated credit
decisions in every year. The Herfindahl--Hirschman index over system shares
stays between $3{,}786$ and $4{,}834$ throughout. Dependence is also
concentrated within institutions, not only
across the market: over a third of reporting lenders routed more than $99\%$ of
their automated decisions through a single provider
(Appendix~\ref{app:hmda}). These data establish shared upstream reliance,
$u_{ig}$, but do not measure historical replay control, $\theta_{ig}$, or
identify model-version changes. They document a structural precondition for
common-mode exposure, not positive uncontrolled exposure
$e_{ig}=u_{ig}(1-\theta_{ig})$ or observed synchronized audit failures.

\section{Experiments on the Governance of the Verifiability Gap}
\label{sec:controlled_experiments}

This section empirically demonstrates how conventional audit indicators can
overstate a firm's retained verification capacity. We isolate the mechanisms
driving the Verifiability Gap across two distinct data regimes:

\begin{itemize}
    \item \textit{Constructed Financial Cases (Studies 1 \& 2):}
    A fixed suite of 32 scenarios across credit, trading,
    anti-money-laundering (AML), and portfolio rebalancing. This controlled
    environment isolates the effects of provider updates and multi-agent
    orchestration.

    \item \textit{Real-World Credit Data (Studies 3 \& 4):}
    The FICO HELOC dataset. We utilize a 1,937-record holdout and a
    32-application subsample to test whether these governance failures persist
    in real credit applications and traditional deterministic models.
\end{itemize}

Across these test beds, our four progressive studies identify complementary
mechanisms through which a system's actual verification capacity can fall behind
its delegated authority:

\begin{description}
    \item[Study 1: Provider Discontinuity.]
    Demonstrates that external endpoint updates and the withdrawal of execution
    controls can sever historical verification even when internal policies
    remain frozen.

    \item[Study 2: Architecture as Policy.]
    Reveals that multi-agent orchestration acts as a latent policy layer. It
    alters financial actions and decouples reproducible outcomes from the
    reproducibility of the underlying decision process.

    \item[Study 3: Historical Replay Failure.]
    Shows that perfectly stable current models can still fail to recover
    historical decisions at policy boundaries, demonstrating that current and
    historical reproducibility are distinct governance quantities.

    \item[Study 4: Reproducibility Without Differentiation.]
    Replicates the orchestration mechanism on real credit applications and
    establishes a governance paradox: perfect current reproducibility can
    coexist with a system that no longer differentiates among cases.
\end{description}

Together, these studies establish the limits of relying on current-state
reproducibility to substantiate historical delegated authority.

\subsection{Reproducibility Is a Governance Profile, Not a Scalar}
\label{subsec:operationalization}

Evaluating reproducibility as a single scalar metric masks critical governance
failures. Let $y_c^0$ and $\mathbf z_c^0$ denote the historical action and
materially relevant process trace recorded for decision $c$ at time $t_0$. Let
$y_{c,k}(t)$ and $\mathbf z_{c,k}(t)$ be their counterparts on repetition $k$
under the environment available at time $t$. $\mathcal C$ is the set of test
episodes, $K$ the repetitions, and $\mathcal L$ the permitted actions. We
decompose reproducibility into four distinct targets:

\textit{Current outcome reproducibility} asks whether today's environment yields
a stable action across repetitions (establishing current stability, not
historical recovery):
\begin{equation}
R_O(t) = \frac{1}{|\mathcal C|}\sum_{c\in\mathcal C}\frac{1}{K}
\max_{\ell\in\mathcal L}\sum_{k=1}^{K}\mathbf 1\!\left[y_{c,k}(t)=\ell\right].
\label{eq:ro_current}
\end{equation}

\textit{Historical outcome reproducibility} asks whether current re-execution
recreates the action recorded at $t_0$. In controlled studies, we use the
baseline modal verdict $\widehat y_c(t_0)$ as a proxy, yielding
\emph{baseline-modal historical reproducibility}:
\begin{equation}
\widetilde R_H(t) = \frac{1}{|\mathcal C|}\sum_{c\in\mathcal C}\mathbf 1\!\left[\widehat y_c(t)=\widehat y_c(t_0)\right].
\label{eq:rh_historical}
\end{equation}

\textit{Material process reproducibility} has a historical and a current form,
and the two answer different governance questions. The historical form asks
whether the materially relevant \emph{past} process can be reconstructed today
to the evidentiary standard $\sigma$:
\begin{equation}
R_P^{H}(t;\sigma) = \frac{1}{|\mathcal C|K}\sum_{c\in\mathcal C}\sum_{k=1}^{K}
\mathbf 1\!\left[\mathbf z_{c,k}(t)\equiv_{\sigma}\mathbf z_c^0\right].
\label{eq:rp_material}
\end{equation}
The current form asks whether the system reproduces its \emph{own} process now,
which is the weaker precondition:
\begin{equation}
R_P^{C}(t;\sigma) = \frac{1}{|\mathcal C|}\sum_{c\in\mathcal C}
\mathbf 1\!\left[\mathbf z_{c,k}(t)\equiv_{\sigma}\mathbf z_{c,k'}(t)
\ \ \forall\, k,k'\right].
\label{eq:rp_current}
\end{equation}
$R_P^{H}$ requires a preserved historical process $\mathbf z_c^0$ to compare
against, so it can only be measured where a historical executable survives.
Study~3 supplies that setting. Study~2 runs entirely in the present and
therefore measures $R_P^{C}$. The two measures are related but not nested: $R_P^{C}<1$ shows that the
system cannot guarantee to repeat its own process, and therefore cannot
guarantee to reconstruct an earlier one; it does not by itself fix the value of
$R_P^{H}$, which must be measured against a preserved historical process.

\textit{Exact trace reproducibility} is a strict diagnostic of computational
stability (requiring identical text sequences), separate from material
equivalence:
\begin{equation}
R_T(t) = \frac{1}{|\mathcal C|}\sum_{c\in\mathcal C}
\mathbf 1\!\left[\mathbf z_{c,1}(t)=\cdots=\mathbf z_{c,K}(t)\right].
\label{eq:rt_identity}
\end{equation}

Finally, high $R_O(t)$ is meaningless if the agent applies a blanket default to
every case. We track \textit{verdict differentiation} $D_V(t)$ as a diagnostic
for degenerate stability. Within each decision family $f$, let $\hat q_f$ be the
distribution of modal verdicts across that family's cases, and $|\mathcal L_f|$
the number of actions the family allows. Then
\begin{equation}
D_V(t) = \frac{1}{|\mathcal F|}\sum_{f\in\mathcal F}
\frac{H\bigl(\hat q_f\bigr)}{\log |\mathcal L_f|},
\label{eq:dv}
\end{equation}
the entropy of those modal verdicts, normalized so that a family answering every
case identically scores zero and one spreading cases evenly across its allowed
actions scores one. The
empirical object is therefore a composite profile, not a scalar:
\begin{equation}
\mathcal R(t;\sigma) =
\bigl(R_O(t),R_H(t),R_P^{H}(t;\sigma),R_P^{C}(t;\sigma),R_T(t);D_V(t)\bigr).
\label{eq:repro_profile}
\end{equation}

This composite profile is a theoretical necessity. As proven in
Appendix~\ref{app:proofs}, $R_O$ is bounded below by $1/|\mathcal L|$ while exact
trace reproducibility ($R_T$) and mutual information ($I(C;Y)$) can decay to
zero. A system can achieve $R_O=1$ by issuing constant defaults without retaining
any case-specific information, rendering scalar outcome metrics fundamentally
insufficient for governance.

\subsection{Experimental Design and Identification Strategy}
\label{subsec:experimental_design}

Studies~1 and~2 utilize 32 fixed financial cases (credit, trading, AML,
portfolio rebalancing) across four isolated testing arms:
\begin{itemize}
    \item \textit{Release-timeline arm ($\widetilde R_H$):} Executes cases across four dated Claude Opus releases over 185 days to isolate provider-driven drift.
    \item \textit{Decoding-control arm ($R_O$):} Varies \texttt{temperature} and \texttt{random seed} on a locally served \texttt{llama3.2:3b} to test execution-control reliance.
    \item \textit{Execution-environment arm ($R_O$):} Compares local vs.\ hosted execution under the tightest available controls to isolate infrastructure effects.
    \item \textit{Orchestration arm ($R_O, R_T, D_V$):} Sweeps configurations from 1 to 50 agents (Figure~\ref{fig:controlled_governance_test}). Information flows strictly serially; only the terminal Decider issues a verdict, cleanly separating trace variation from terminal outputs.
\end{itemize}

Agent count and chain depth are distinct quantities, and it is the second that
the theory constrains. Each configuration places three specialists per layer, so
adding agents widens the graph faster than it lengthens it: the largest
configuration has 52 nodes but only 19 sequential stages
(Figure~\ref{fig:topology}). The stage count is the $k$ that enters the
contraction bound, which is why the reported sweep sizes should not be read as
chain lengths. Adding agents also changes width, call count, and context length
at the same time, so the sweep does not isolate seriality on its own. What the
theory does fix is the direction. Lemma~\ref{lem:book}$(i)$ gives pairwise
trace agreement $A_T=\prod_{i=1}^{m}p_i$ for a fixed execution process, so
under $p_i\le\bar p<1$ it can only fall as calls are added; comparing
redesigned architectures requires the same bound to hold across them. The sweep tests whether that
worst case is what a governance function actually faces.

\begin{figure}[!htbp]
\centering
\includegraphics[width=0.78\linewidth]{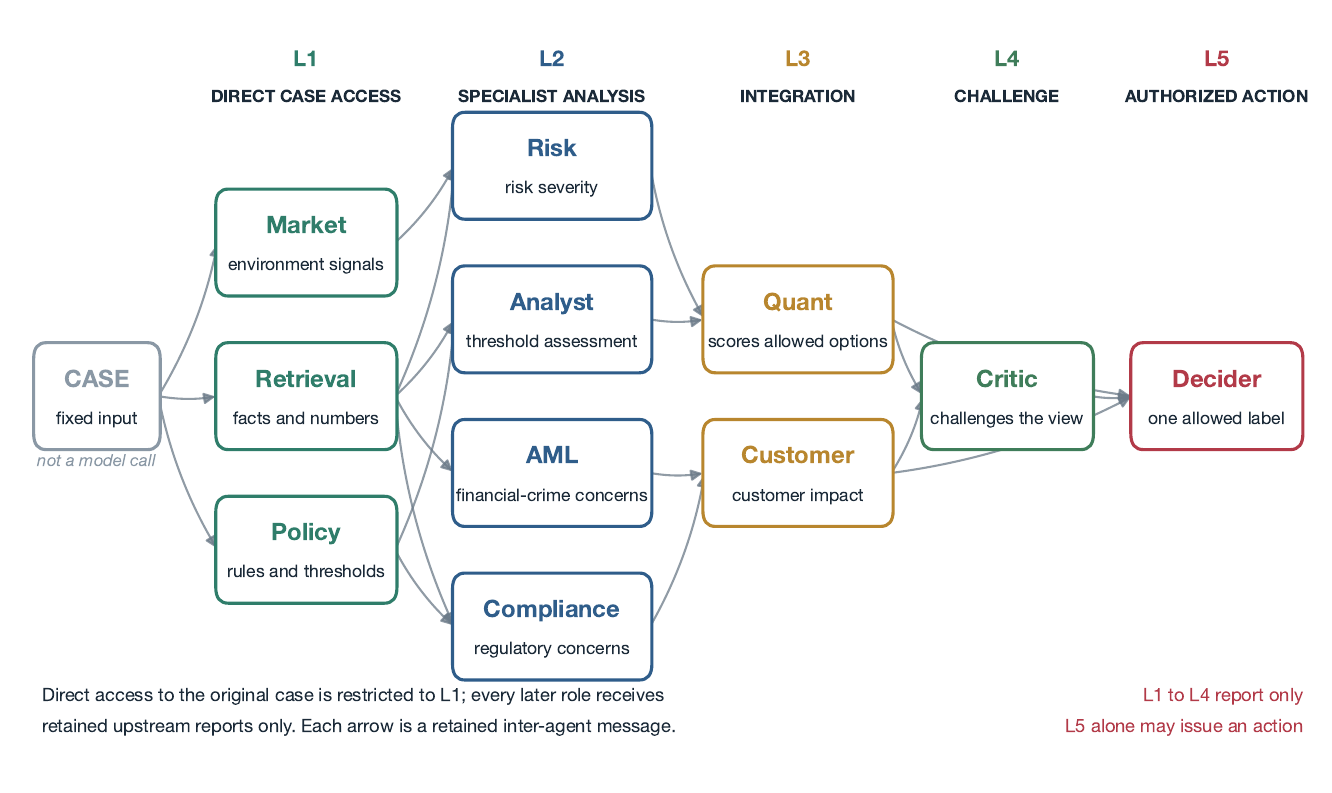}
\caption{Controlled multi-agent governance test, drawn from the role graph in
the replication package. Only Retrieval, Market, and Policy see the case;
every later role sees only the reports transmitted along the designated edges.
Roles in L1--L4 are report-only and the terminal Decider alone issues a verdict.
Every role output and handoff is retained, which is what makes the first
divergence locatable.}
\label{fig:controlled_governance_test}
\end{figure}

\begin{figure}[!htbp]
\centering
\includegraphics[width=0.84\linewidth]{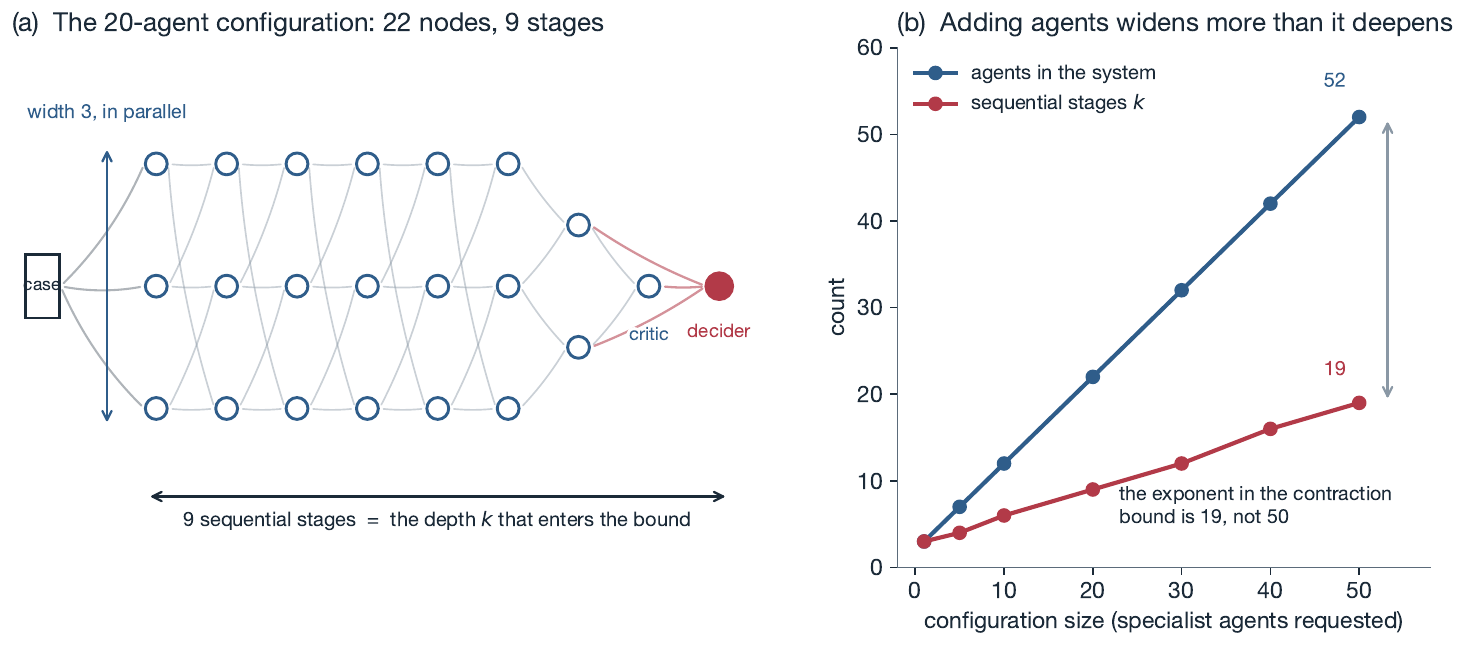}
\caption{Agent count and chain depth are different quantities.
\textbf{(a)} The twenty-agent configuration: three specialists per layer, each
reading two reports from the layer above and never the case itself. Twenty-two
nodes span nine sequential stages.
\textbf{(b)} Node count grows linearly with configuration size; sequential
stages grow far more slowly. The largest configuration has 52 nodes but 19
stages, and the stage count is what enters the contraction bound.}
\label{fig:topology}
\end{figure}

\subsection{Study 1: Provider Dominance Over the Governance Environment}
\label{subsec:controlled_evidence}

To strictly isolate provider-side exposure, Study~1 holds the 32 financial cases
and systemic instructions constant. It varies only those dimensions exclusively
governed by the AI provider: the specific release version that executes the
decision, and the algorithmic controls permitted for subsequent replay. These
are two distinct ways a hosted service can break historical verification. The
provider can replace the executable behind an endpoint, which changes the object
of audit; or it can withdraw the controls a replay requires, which changes the
means of verification. Three arms separate them. The firm may retain its prompt
and its output and still be unable to recreate the historical environment.

\subsubsection{Release-Timeline Arm: Changing the Object of Audit}
Across four Claude Opus releases spanning 185 days, with three repetitions per
case and release, baseline-modal historical reproducibility ($\widetilde R_H$) held at
1.000 through release two, then fell to 0.906
(Figure~\ref{fig:provider_audit_surface}a). The measure is deliberately
conservative: it asks only whether the action category is recovered, not whether
the full trace or the explanation is. Every divergent verdict shifted
conservatively (e.g., \emph{buy} to \emph{hold}). This constitutes a
cross-version shift in operative risk posture under unchanged case and policy
inputs. A release
may bundle changes to behavior, capability, and infrastructure, and we do not
attribute an observed shift to any particular undocumented internal change.
Moreover, aggregate stability masked case churn: the cases deviating in release three
differed from those in release four, meaning portfolio-level stability cannot
pinpoint which historical decisions require audit.

More critically, Anthropic deprecated institutional control over reproduction
mechanisms. Requests to \texttt{claude-opus-5} with \texttt{temperature=0} now
return \texttt{HTTP 400}: ``\texttt{`temperature` is deprecated}''. Random seeds
were never exposed. This exemplifies \emph{control-surface dependence}: retained
reproducibility relies on execution controls the institution neither owns nor
can preserve. A fixed seed cannot compensate once the underlying model or tool
context shifts (Figure~\ref{fig:provider_audit_surface}b).

\begin{figure}[!htbp]
\centering

\begin{subfigure}[t]{0.42\linewidth}
\centering
\vspace{0pt}
\includegraphics[width=\linewidth]{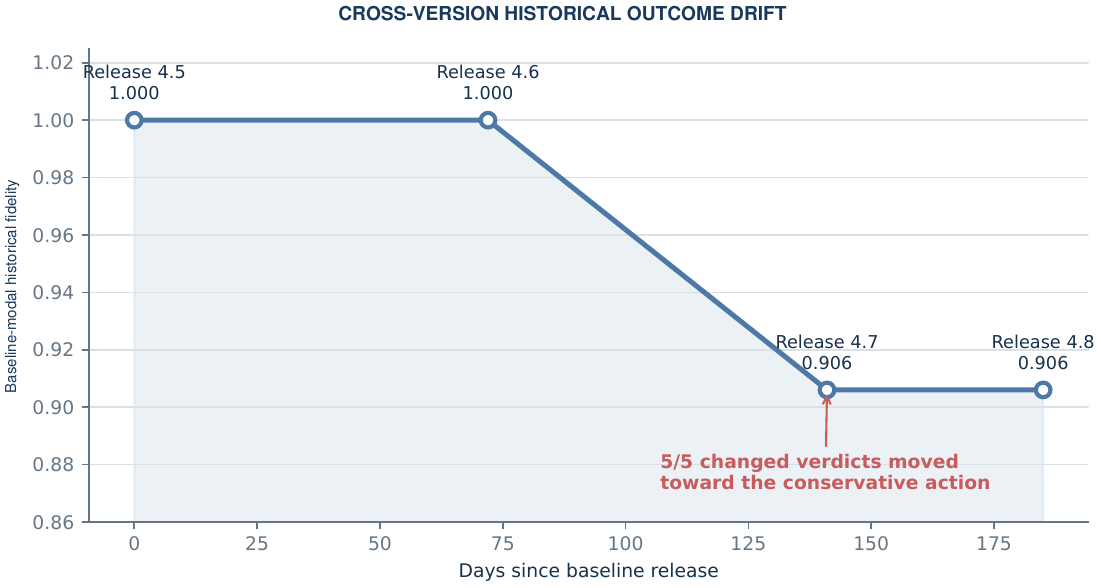}
\caption{Release-timeline drift.}
\label{fig:versions}
\end{subfigure}
\hfill
\begin{subfigure}[t]{0.55\linewidth}
\centering
\vspace{0pt}
\includegraphics[width=\linewidth]{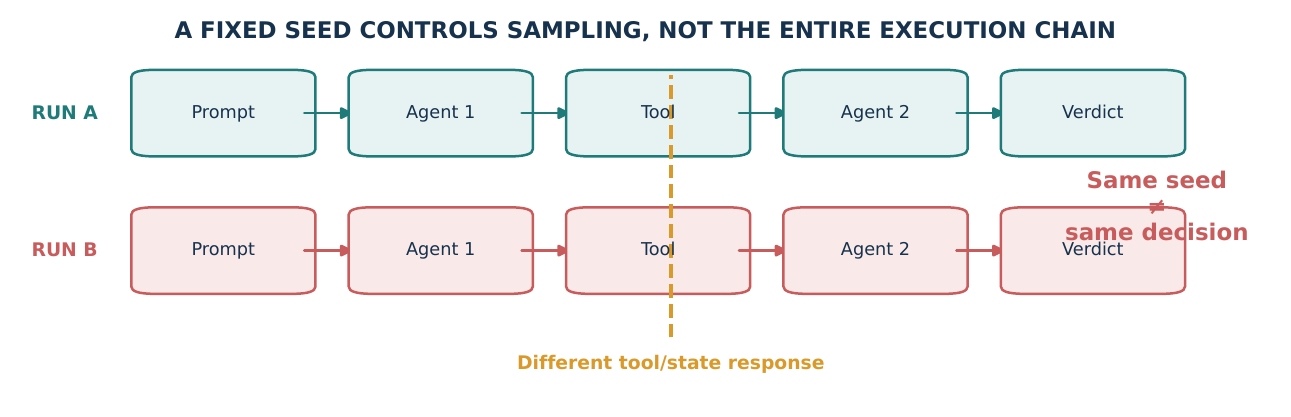}
\caption{Limits of seed-based replay.}
\label{fig:seed_chain}
\end{subfigure}

\caption{Provider dominance over the audit surface.
\textbf{(a)} Baseline-modal historical drift across four dated Claude Opus
releases: $\widetilde R_H$ holds at 1.000 through the second release, then falls
to 0.906. Every changed verdict moves toward a more conservative action.
Releases also alter capability, so this is release-timeline drift, not silent
same-version drift.
\textbf{(b)} A seed controls sampling, not the execution chain. It cannot
recreate a decision once model state, tool responses, or context have drifted.}
\label{fig:provider_audit_surface}
\end{figure}

These results should not be read as a comparison of model quality. A changed
action may be more accurate, more conservative, or better aligned with a new
provider policy. Historical reproducibility asks a different question: can the current
environment recover what the institution actually did before? Even a beneficial
update creates an audit discontinuity when the earlier executable and control
surface are gone.

\subsubsection{Decoding-Control Arm: Changing the Means of Replay}
Panel~A of Table~\ref{tab:decodinggrid} isolates the exposed execution controls
on a locally served model, where the case and policy logic are held fixed. With
temperature pinned to zero and the version fixed, both the single-agent and the
ten-agent configurations reproduced every modal action across five repetitions.
When temperature was not sent and seeds varied, current outcome reproducibility
fell to 0.912 for one agent and 0.794 for ten; when neither temperature nor seed
was supplied, to 0.831 and 0.744. Replay therefore depends on controls that sit
outside the financial case entirely. The governance problem is not that
stochastic decoding is improper. It is that a firm cannot promise later
reconstruction when it neither retains the historical setting nor decides
whether the provider keeps exposing it.

\subsubsection{Execution-Environment Arm: Tight Controls Do Not Guarantee Replay}
Panel~B of Table~\ref{tab:decodinggrid} demonstrates that reproducibility
requires absolute execution control. Locally, with temperature pinned to zero,
\texttt{llama3.2:3b} returned the modal verdict in 320/320 executions.
Conversely, a hosted \texttt{claude-haiku} model pinned to a dated release
deviated in 1 of 320 executions. When controls were entirely removed
(\texttt{claude-opus-5}, provider default), the model deviated in 1 of 960
executions. Notably, this deviation flipped a clear credit approval to a
\emph{refer}, offering a highly defensible alternative rationale based on a
single delinquency. Irreproducibility on hosted endpoints is therefore not mere
noise; it is the stochastic injection of defensible alternative judgments beyond
institutional control.

We report these as counterexamples rather than estimates. One deviating
execution does not support a rate, a difference of one case in 32 is not
distinguishable from chance on this sample, and the three arms differ in
repetition count and available controls, so their execution-level figures are
not three measurements of one quantity. What the deviations establish is
narrower and sufficient: reproducibility is not guaranteed on a hosted endpoint,
whether every exposed control is pinned or none is exposed at all.

\begin{table}[!ht]
\centering
\caption{\text{Execution controls and current outcome reproducibility.} Across 32 financial cases, Panel~A varies decoding parameters on a fixed local model (five repetitions), while Panel~B compares execution environments under each endpoint's tightest available replay controls (identical prompts; commercial deliberation disabled across all arms).}
\label{tab:decodinggrid}
\footnotesize
\setlength{\tabcolsep}{5pt}
\renewcommand{\arraystretch}{1.02}

\textbf{Panel A. Decoding controls, local model}\\[2pt]
\begin{tabular}{@{}llrr@{}}
\toprule
\textbf{Temperature} & \textbf{Seed} & \textbf{1 agent} & \textbf{10 agents} \\
\midrule
Pinned to 0 & Varied & 1.000 & 1.000 \\
Not sent & Varied & 0.912 & 0.794 \\
Not sent & Not sent & 0.831 & 0.744 \\
\bottomrule
\end{tabular}

\vspace{8pt}
\textbf{Panel B. Execution environment, tightest available controls}\\[2pt]
\begin{tabular}{@{}>{\raggedright\arraybackslash}p{2.5cm}>{\raggedright\arraybackslash}p{3.0cm}>{\raggedright\arraybackslash}p{4.3cm}rr@{}}
\toprule
\textbf{Environment} & \textbf{Model} & \textbf{Replay controls exposed} & \textbf{Modal} & \textbf{Cases} \\
 & & & \textbf{verdict} & \textbf{differing} \\
\midrule
Local, single stream & \texttt{llama3.2:3b} & temperature 0; version pinned & 320 / 320 & 0 / 32 \\
Hosted, shared & \texttt{claude-haiku-4-5} & temperature 0; dated release & 319 / 320 & 1 / 32 \\
Hosted, shared & \texttt{claude-opus-5} & provider-default decoding; fixed model ID & 959 / 960 & 1 / 32 \\
\bottomrule
\end{tabular}
\end{table}

Study~1 therefore identifies two separate governance shocks. A release can change
the action recovered from the same case, and an endpoint can remove the controls
required to recreate the prior execution. These are not ordinary performance
drifts; they change the historical object against which validation and
accountability operate. The organizational response they call for is
change-triggered reauthorization of delegated authority, together with
preservation of an executable evidence bundle wherever that is technically and
contractually possible.

\subsection{Study 2: Orchestration, Serial Dependence, and Output Collapse}
\label{subsec:orchestration_experiment}

Study~2 holds the local model and inputs fixed, altering only the orchestration
architecture.

\subsubsection{Architecture Is Policy and Trace--Outcome Decoupling}
This first comparison fixes temperature at zero. That isolates architecture from
decoding variation. Under this control each architecture reproduced its own
modal verdict in every repetition. $R_O$ is therefore 1.000 for the one-agent
and the ten-agent configuration alike (Table~\ref{tab:decodinggrid}, Panel~A).
Internal reproducibility was perfect on both sides. The two architectures still
disagreed with each other. Their modal verdicts matched on only 16 of 32 cases.
Orchestration forms a latent policy layer that alters the operative risk posture
without explicit mandate. Because the designed roster partly conditions the
direction of that shift, the evidence supports a structural effect of
architecture rather than a universal claim that multi-agent systems become more
conservative.

At the provider default, a separate ten-repetition run gives the
family-level view in Figure~\ref{fig:agents_policy_trace}a: mean $R_O$ falls
from 0.903 at one agent to 0.747 at ten. The five-repetition sweep in
Section~\ref{subsubsec:degenerate} reports 0.912 and 0.794 for its own
configurations; the two runs are not directly comparable. The trace behaves differently again.
Under that condition the recorded trace differed across repetitions at every
stage in all 32 cases, yet the final verdict differed in only 25
(Figure~\ref{fig:agents_policy_trace}b). This establishes
\emph{trace--outcome decoupling}: downstream agents absorb upstream linguistic
variation, meaning a stable final verdict ($R_O$) offers no proof that the
material decision process ($R_P^{C}$) was reproduced.

That inference requires measuring $R_P^{C}$ directly. We operationalize
current process agreement using a lexical-content proxy for $\equiv_\sigma$:
corresponding stage outputs match when their multisets of numeric quantities
and sets of allowable-action mentions match. This rule ignores wording and
order, but it does not resolve variable attribution, units, or negation. A
case passes only when all corresponding stages match across repetitions. The
proxy therefore measures agreement under the specified extraction rule, not
independently validated semantic equivalence.

Under this standard the single-call configuration retains partial material
reproducibility, $R_P^{C}=0.250$ on the three-billion-parameter model and $0.094$ on
the eight-billion one, while every multi-agent configuration returns
$R_P^{C}=0$ (Table~\ref{tab:rp}). Under a second, looser rule, so that numeric quantities are ignored entirely and
only the recommended actions must match, raises single-call $R_P^{C}$ to $0.562$ on
both models and still leaves $0.031$ and $0.000$ at five agents. 

Because the measure is a conjunction across stages and repetitions, adding
stages makes it mechanically harder to satisfy. This also applies to the
decline between one and five agents. The observed scores establish
instability under the specified proxy; they do not identify how much of the
decline reflects material process change rather than additional comparisons.
Nor does a zero score establish that the retained historical trace is
unreconstructable. Outcome agreement cannot substitute for a separate
assessment of material process equivalence.

\begin{table}[!ht]
\centering
\caption{\textbf{Current process agreement under the specified numeric/action proxy ($R_P^C$) and exact trace reproducibility ($R_T$).} Share of 32 cases with mutually equivalent repetitions under $\equiv_\sigma$ (Equation~\ref{eq:rp_current}; distinct from historical replay in Study~3). Frontier arm tested at three configurations only (n/a: not run).}
\label{tab:rp}
\footnotesize
\renewcommand{\arraystretch}{1.02}
\begin{tabular}{@{}lrrrrrrr@{}}
\toprule
\textbf{Agents} & \textbf{1} & \textbf{5} & \textbf{10} & \textbf{20} & \textbf{30} & \textbf{40} & \textbf{50} \\
\midrule
$R_P^{C}$, \texttt{llama3.2:3b} & 0.250 & 0.000 & 0.000 & 0.000 & 0.000 & 0.000 & 0.000 \\
$R_P^{C}$, \texttt{llama3.1:8b} & 0.094 & 0.000 & 0.000 & 0.000 & 0.000 & 0.000 & 0.000 \\
$R_P^{C}$, frontier commercial  & 0.250 & n/a   & 0.000 & n/a   & n/a   & 0.000 & n/a   \\
\midrule
$R_T$, all three scales     & 0.000 & 0.000 & 0.000 & 0.000 & 0.000 & 0.000 & 0.000 \\
\bottomrule
\end{tabular}
\end{table}

\begin{figure}[!htbp]
\centering

\begin{subfigure}[t]{0.48\linewidth}
\centering
\includegraphics[width=\linewidth]{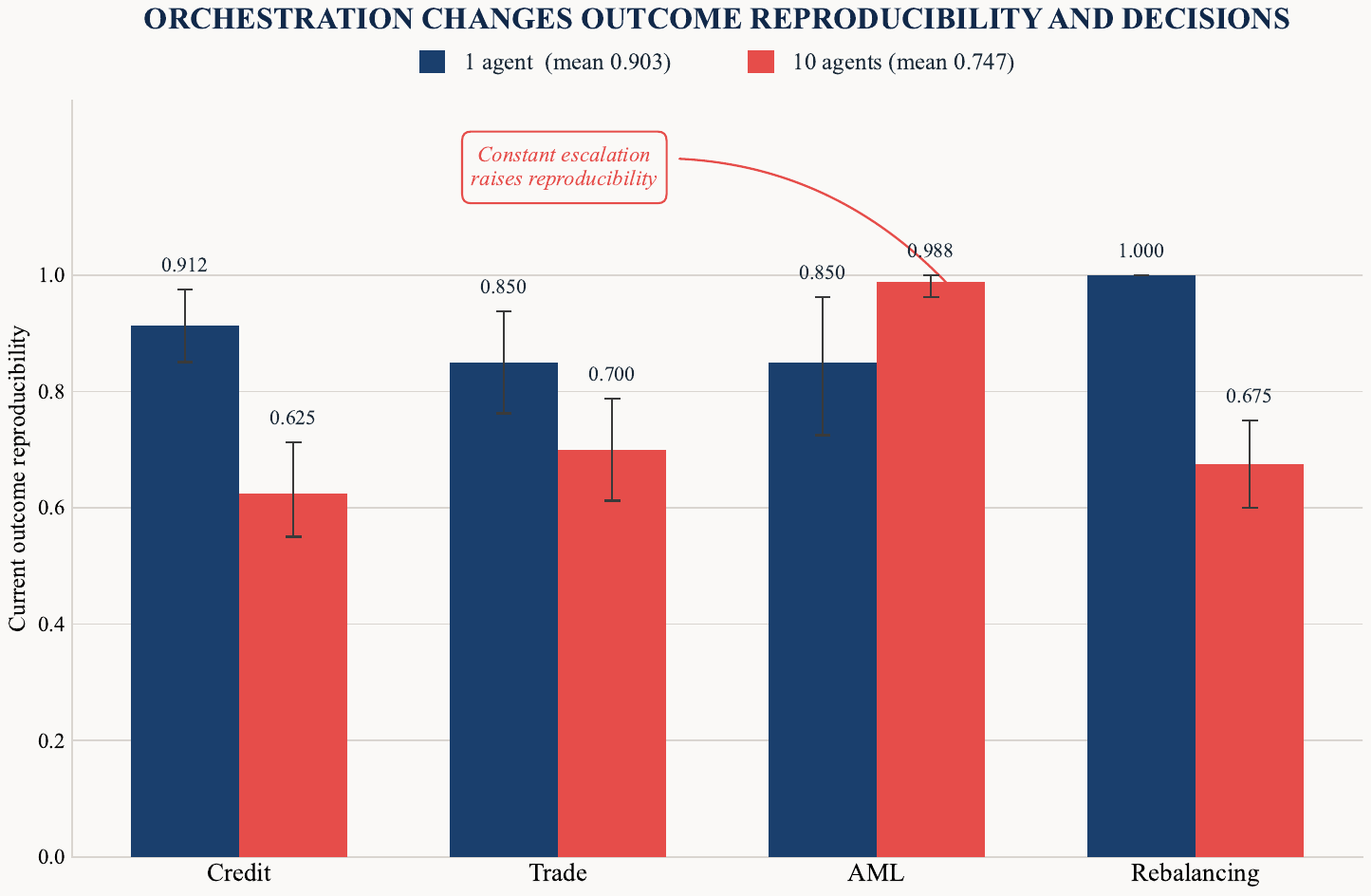}
\caption{Architecture-level reproducibility by decision family, with 95\%
confidence intervals over the eight cases in each family. Axes are zero-based.}
\label{fig:agentsfid}
\end{subfigure}
\hfill
\begin{subfigure}[t]{0.48\linewidth}
\centering
\includegraphics[width=\linewidth]{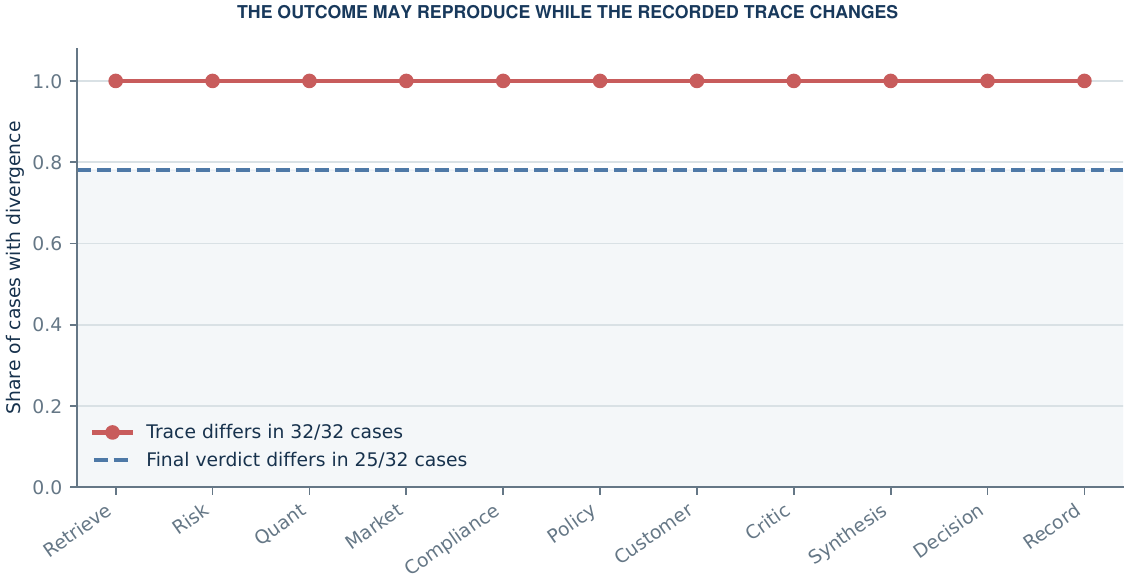}
\caption{Trace--outcome decoupling.}
\label{fig:agentsstage}
\end{subfigure}

\caption{Orchestration changes both the decisions reproduced and the records
supporting them.
\textbf{(a)} Mean $R_O$ falls from 0.903 at one agent to 0.747 at ten. AML moves
the other way because near-constant escalation makes the verdict artificially
easy to reproduce.
\textbf{(b)} Traces differ in all 32 cases at every stage; the verdict differs
in only 25. This establishes exact trace nonidentity. Material process
equivalence requires a separate test under $\equiv_{\sigma}$.}
\label{fig:agents_policy_trace}
\end{figure}

\subsubsection{Calibrating the Theory Against the Sweep}

Table~\ref{tab:calib} places the measured quantities beside the ones the theory
constrains, using the pairwise statistics $A_T$ and $A_O$ of
Lemma~\ref{lem:book}. Three readings follow. First, pairwise trace agreement
$A_T$ is $0.0000$ in every released-temperature row at both scales, so by
Lemma~\ref{lem:book}$(ii)$ the outcome metric equals the compression
probability $\kappa$. Since $\kappa$ is estimated from the same pairs, this
is a decomposition identity and a consistency check, not an independent test;
its content is that every movement in outcome agreement across the sweep is
movement in $\kappa$ alone. Second, case information is lower in some larger configurations, but the
pattern is non-monotonic. On the smaller model, $I(C;Y)$ falls from $0.2986$
bits at one agent to $0.0853$ at forty, then returns to $0.1990$ at fifty; on
the larger model it moves from $0.2434$ at one to $0.1105$ at fifty. Because
configurations also differ in width and context, these contrasts do not
identify serial compression alone. Third, the one configuration in which both terms are visible is the
eight-billion model at temperature zero with ten agents, where $A_T=0.4500$ and
$\kappa=0.8409$ are simultaneously interior; everywhere else the process term has
already reached its floor.

\begin{table}[!ht]
\centering
\caption{Calibration of the theory against the sweep. \emph{Released temperature}
represents provider withdrawal of temperature control. $A_T$ is pairwise trace
agreement, not the all-repetition share $R_T$. Pairwise outcome
agreement $A_O=\Pr[Y^{(1)}=Y^{(2)}]$ is linked to compression probability
$\kappa$ by Lemma~\ref{lem:book} and cannot exceed the modal share
(Equation~\ref{eq:ro_current}). Mutual information uses a uniform case prior
within each decision family. With $L=2$ or $3$ actions, the agreement floor
is $1/L$ ($.33$--$.50$), even when process evidence is lost (n/a: undefined).}
\label{tab:calib}
\footnotesize
\begin{tabular}{@{}llrrrr@{}}
\toprule
\textbf{Model} & \textbf{Condition} & \textbf{$A_T$} & \textbf{$A_O$} &
\textbf{$\kappa$} & \textbf{$I(C;Y)$ (bits)} \\
\midrule
\texttt{llama3.2:3b} & released, $n=1$   & 0.0000 & 0.8313 & 0.8313 & 0.2986 \\
                     & released, $n=5$   & 0.0000 & 0.5938 & 0.5938 & 0.2130 \\
                     & released, $n=10$  & 0.0000 & 0.6594 & 0.6594 & 0.2118 \\
                     & released, $n=20$  & 0.0000 & 0.7844 & 0.7844 & 0.1429 \\
                     & released, $n=30$  & 0.0000 & 0.6687 & 0.6687 & 0.1781 \\
                     & released, $n=40$  & 0.0000 & 0.7875 & 0.7875 & 0.0853 \\
                     & released, $n=50$  & 0.0000 & 0.7344 & 0.7344 & 0.1990 \\
\midrule
\texttt{llama3.2:3b} & temperature 0, $n=1$  & 1.0000 & 1.0000 & n/a    & n/a \\
                     & temperature 0, $n=10$ & 0.9375 & 1.0000 & 1.0000 & n/a \\
\midrule
\texttt{llama3.1:8b} & released, $n=1$   & 0.0000 & 0.7875 & 0.7875 & 0.2434 \\
                     & released, $n=5$   & 0.0000 & 0.5813 & 0.5813 & 0.2474 \\
                     & released, $n=10$  & 0.0000 & 0.6281 & 0.6281 & 0.2591 \\
                     & released, $n=20$  & 0.0000 & 0.5375 & 0.5375 & 0.1076 \\
                     & released, $n=30$  & 0.0000 & 0.5281 & 0.5281 & 0.1366 \\
                     & released, $n=40$  & 0.0000 & 0.6562 & 0.6562 & 0.1494 \\
                     & released, $n=50$  & 0.0000 & 0.5312 & 0.5312 & 0.1105 \\
\midrule
\texttt{llama3.1:8b} & temperature 0, $n=1$  & 0.9688 & 1.0000 & 1.0000 & n/a \\
                     & temperature 0, $n=10$ & 0.4500 & 0.9125 & 0.8409 & n/a \\
\bottomrule
\end{tabular}
\end{table}

\subsubsection{Degenerate Outcome Reproducibility}
\label{subsubsec:degenerate}
Scaling from one to fifty agents revealed a non-monotonic $R_O$
(Figure~\ref{fig:agentsweep}), dropping to 0.750 at five agents before
artificially recovering to 0.875 at forty. This recovery is a governance
illusion. Exact trace reproducibility ($R_T$) was strictly zero across all 1,120
executions (Table~\ref{tab:agentsweep}).

The apparent $R_O$ recovery coincides with \emph{categorical
compression}. As pipeline depth increases, agents apply blanket actions,
collapsing verdict differentiation ($D_V$) from 0.324 to 0.086. For example, AML
cases initially split between \emph{clear} and \emph{escalate}; by forty agents,
all eight cases were escalated. We term this \emph{degenerate outcome
reproducibility}: metrics signal robust stability precisely when the system has
lost the capacity to distinguish cases.

This mechanism scales. Replicating the sweep on a larger \texttt{llama3.1:8b}
model yielded the same dynamic in direction: $R_T=0$, with $R_O$ recovering and
$D_V$ falling toward forty agents (Figure~\ref{fig:scale_compare}). For the
eight-billion model the decline is directional rather than resolved, because its
intervals at five and at forty agents overlap; only the three-billion collapse
separates cleanly from its own earlier values. While the larger
model retained more differentiation (0.214 vs.\ 0.086), its outcome
reproducibility was systematically lower beyond ten agents. A frontier commercial model
(\texttt{claude-sonnet-5}) run at three configurations is consistent with the
same reading. It starts from a higher differentiation ($D_V=0.749$). Its exact
trace reproducibility is still zero, and its differentiation still falls by
68\% at extreme scale.
Capability improves outcome reproducibility; in the configurations observed it
does not guarantee auditability.

Two cautions bound this comparison. The provider does not disclose the frontier
model's parameter count, so scale is reported as a tier rather than a continuous
covariate, and the three endpoints do not share a decoding condition, because
temperature is settable locally and rejected by the commercial endpoint. And
$D_V$ is unstable at these sample sizes: recomputing the eight-billion-parameter
model on three repetitions moves its forty-agent value from 0.214 to 0.414. We
therefore report change within a model and avoid ranking models at a single
configuration.

\begin{table}[!ht]
\centering
\caption{\textit{Agent scale versus reproducibility and verdict differentiation.} Evaluated on locally served \texttt{llama3.2:3b} across 32 cases (five repetitions, unconstrained temperature, varied seeds). The single-agent row serves as the baseline for the underlying model under identical execution settings.}
\label{tab:agentsweep}
\footnotesize
\renewcommand{\arraystretch}{1.02}
\begin{tabular}{@{}rrrrr@{}}
\toprule
\textbf{Agents} & \textbf{$R_O$} & \textbf{Exact trace} & \textbf{Cases with a} & \textbf{$D_V$} \\
 & & \textbf{reproducibility} & \textbf{verdict split} & \\
\midrule
1  & 0.912 & 0/32 & 13/32 & 0.324 \\
5  & 0.750 & 0/32 & 25/32 & 0.336 \\
10 & 0.794 & 0/32 & 21/32 & 0.367 \\
20 & 0.869 & 0/32 & 13/32 & 0.264 \\
30 & 0.800 & 0/32 & 21/32 & 0.286 \\
40 & 0.875 & 0/32 & 14/32 & 0.086 \\
50 & 0.850 & 0/32 & 18/32 & 0.086 \\
\bottomrule
\end{tabular}
\end{table}

\begin{figure}[!htbp]
\centering
\includegraphics[width=0.79\linewidth]{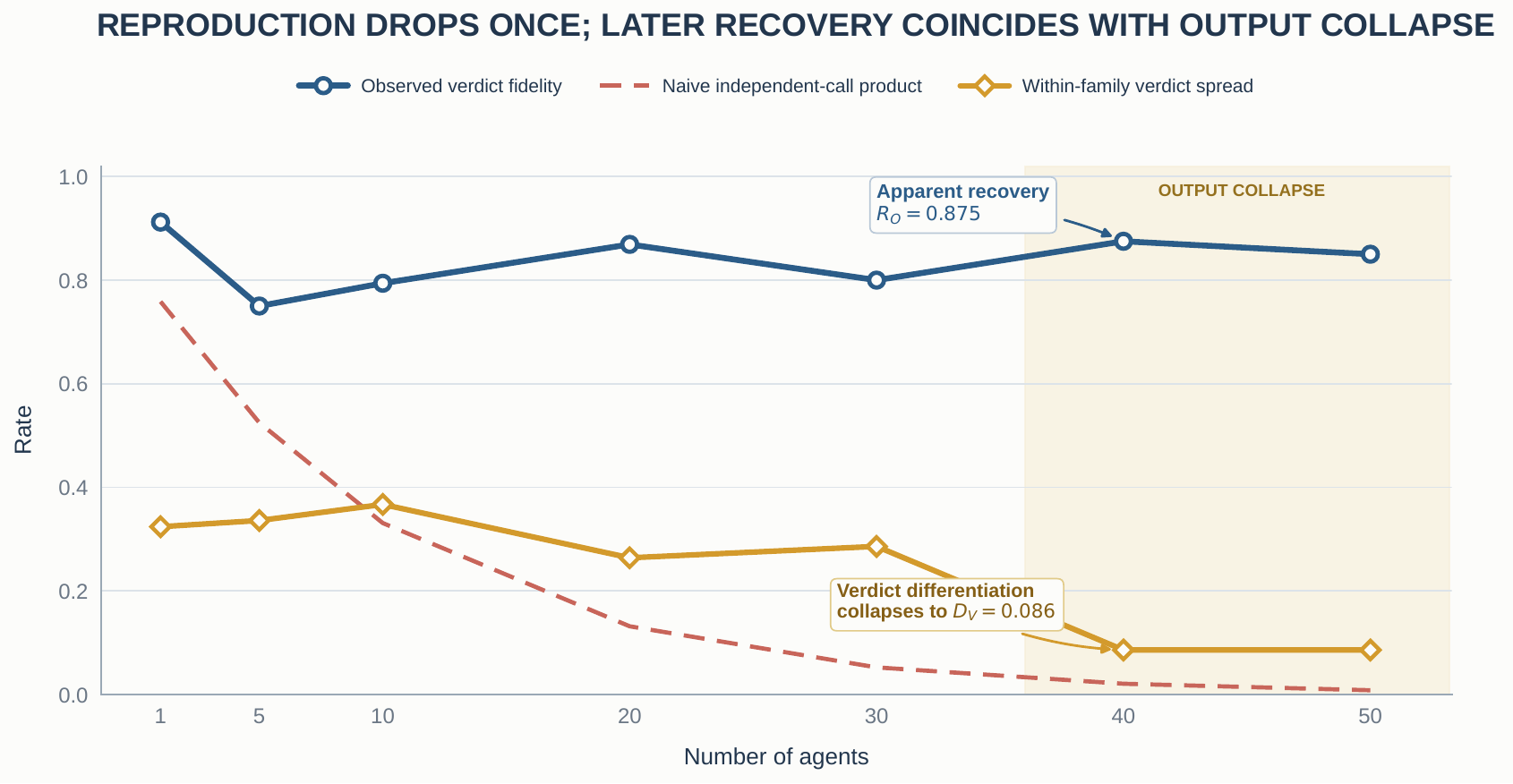}
\caption{\textbf{Degenerate outcome recovery.} Pipeline depth initially reduces $R_O$ via trace variance, but the rebound coincides with serial summarization mapping divergent traces onto default actions ($R_O$ rebounds as $D_V$ collapses). The dashed benchmark assumes independent error propagation, omitting this many-to-one compression.}
\label{fig:agentsweep}
\end{figure}

\begin{figure}[!htbp]
\centering
\includegraphics[width=0.86\linewidth]{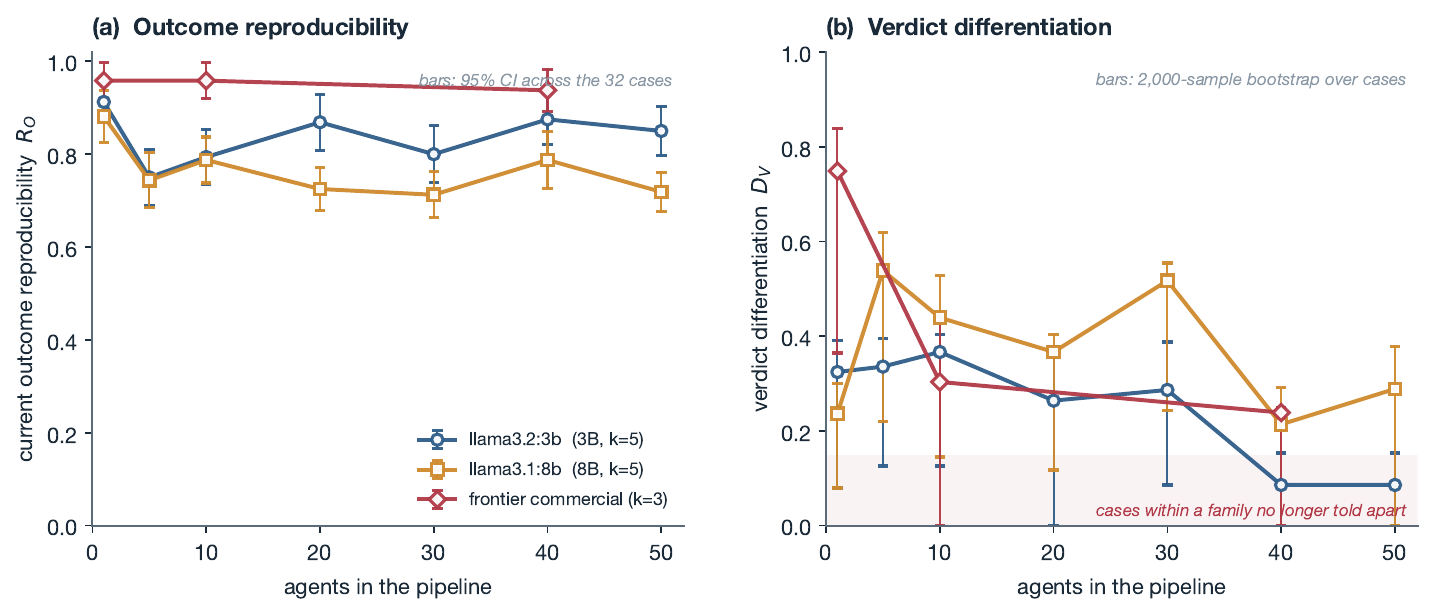}
\caption{\textit{Agent scaling across model tiers (95\% CIs).}
\textbf{(a)}~$R_O$: Local models dip at 5 agents then rebound (8B underperforms 3B beyond 10 agents); frontier $R_O$ remains strictly superior without CI overlap.
\textbf{(b)}~$D_V$: Reflects $4\times 8$ sample size; only 3B at 40 agents separates cleanly from baseline, though frontier differentiation falls steepest.
$R_T = 0$ throughout. (Frontier: 3 configs, 3 reps; local: 7 configs, 5 reps).}
\label{fig:scale_compare}
\end{figure}

\subsection{Study 3: A Reproducible Credit Model Can Fail to Reproduce History}
\label{subsec:credit_experiment}

Studies~1 and~2 demonstrate severe verifiability deficits, yet both rely on
stochastic generation and complex serial orchestration. Consequently, neither
definitively isolates the failure from these systemic complexities. To prove that
historical reproducibility failures require neither generative stochasticity nor
multi-agent handoffs, Study~3 deploys a fully deterministic, single-step credit
model. By stripping away these confounders, we isolate the historical
reproducibility failure that a version change alone produces.

Using the FICO HELOC dataset \citep{fico2018challenge}, version $v_1$ (5,808
records) and updated version $v_2$ (7,745 records) evaluate a strict
1,937-record holdout. Both are standardized logistic regressions on the same five variables
(\texttt{StandardScaler} + \texttt{LogisticRegression}, seed 2026); a
stratified 80/20 split with the same seed sets aside the holdout, $v_1$ is
fitted on 75\% of the development set and $v_2$ on all of it, so $v_2$ is a
refresh with more data and unchanged code and thresholds.

At the portfolio level, aggregate validation masks case-level divergence
($\mathrm{AUC}_{v_1}=.7901$ vs.\ $\mathrm{AUC}_{v_2}=.7899$). Yet, across the
holdout, 23 actions changed, yielding $R_H = .9881$. Consider Case 2270: $v_1$
estimates $p=.5580$ (\emph{deny}), while $v_2$ estimates $p=.5488$
(\emph{refer}). Across repeated executions, both frozen versions exhibit perfect
current stability:
\begin{equation}
R_O(v_1)=R_O(v_2)=1, \qquad R_H(\text{case }2270)=0.
\label{eq:credit_ro_rh_separation}
\end{equation}

The 23 discordant cases share a geometry: an action changes only when the score
movement exceeds the distance to the nearest cutoff, and every one of them
begins within .0161 of a threshold while moving no further than a typical case
(Figure~\ref{fig:study3_robust}a). Nor is the result an artifact of the chosen
thresholds, since all 39 admissible cutoff pairs on a .05 grid change at least 14
held-out actions (Figure~\ref{fig:study3_robust}b).

This establishes a critical separation: $R_O=1$ can coexist with $R_H=0$ at the
decision boundary. If a firm retains only current model $v_2$, attempting to
audit historical Case 2270 will yield a \emph{refer}, opening a decision-level
Verifiability Gap (Figure~\ref{fig:credit_verifiability_gap}). An explanation
generated from $v_2$ may accurately describe current logic, but it is
evidentially nonresponsive to the historical denial.

Conversely, preserving the $v_1$ executable restores verification capacity,
allowing auditors to perfectly reproduce the original $p=.5580$ denial and
conduct causal faithfulness probes (Table~\ref{tab:credit_experiment}). The
relevant audit object is never merely the current system; it is the specific
historical executable that exercised the delegated authority.

\begin{table}[!ht]
\centering
\caption{Credit-decision experiment: aggregate stability can conceal
decision-level historical-reproducibility failure.}
\label{tab:credit_experiment}
\footnotesize
\renewcommand{\arraystretch}{1.02}
\begin{tabular}{@{}p{0.22\linewidth}p{0.37\linewidth}p{0.33\linewidth}@{}}
\toprule
\textbf{Evidence test}
& \textbf{Observed result}
& \textbf{Governance implication} \\
\midrule

Version-level validation
& $\mathrm{AUC}_{v_1}=.7901$ and
  $\mathrm{AUC}_{v_2}=.7899$; 23 of 1{,}937 held-out actions change, yielding
  $R_H=.9881$.
& A stable aggregate AUC does not guarantee historical reproducibility. \\

Case-level action
& $v_1$: $p=.5580$, deny; $v_2$: $p=.5488$, refer. A score movement of .0092
  crosses the fixed denial threshold.
& Small numerical changes can produce discrete and consequential financial
  action changes at policy boundaries. \\

Current reproducibility
& Each frozen version returns its own score and action in 10/10 executions.
& $R_O=1$ can coexist with case-level $R_H=0$ because current reproducibility
  and historical reproducibility evaluate different versions. \\

Faithfulness probes
& Three one-factor interventions on $v_1$ lower the score below .55 and change
  deny to refer exactly as the stated reasons predict.
& Reasons have evidentiary value only when tested against the executable that
  produced the historical action. \\

Retention condition
& Current-only replay returns the $v_2$ referral; replay of the preserved $v_1$
  evidence bundle recovers the historical denial and permits faithfulness
  testing.
& Current-only retention opens the Verifiability Gap; preserving the
  historical executable restores verification capacity. \\

\bottomrule
\end{tabular}
\end{table}

\begin{figure}[!htbp]
\centering
\includegraphics[width=0.86\linewidth]{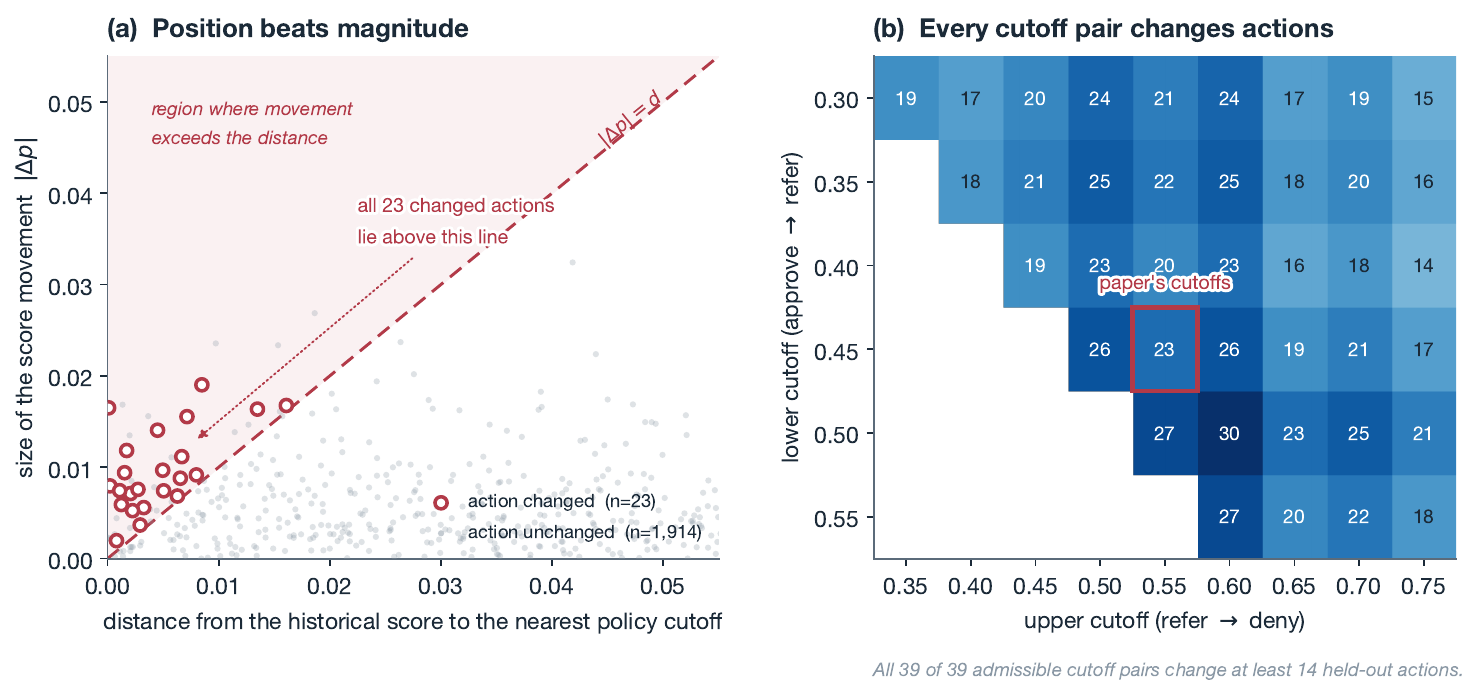}
\caption{\textit{Action changes under model refresh.}
\textbf{(a)} Score movement ($|\Delta p|$) vs.\ distance to nearest $v_1$ cutoff ($d$). Actions change strictly above the identity line $|\Delta p| = d$; all 23 flipped cases lie above it, driven by cutoff proximity rather than shift magnitude.
\textbf{(b)} All 39 admissible threshold pairs on a 0.05 grid alter $\ge 14$ actions (outlined baseline: 23; median: 21).}
\label{fig:study3_robust}
\end{figure}

\begin{figure}[!htbp]
\centering
\includegraphics[width=0.82\linewidth]{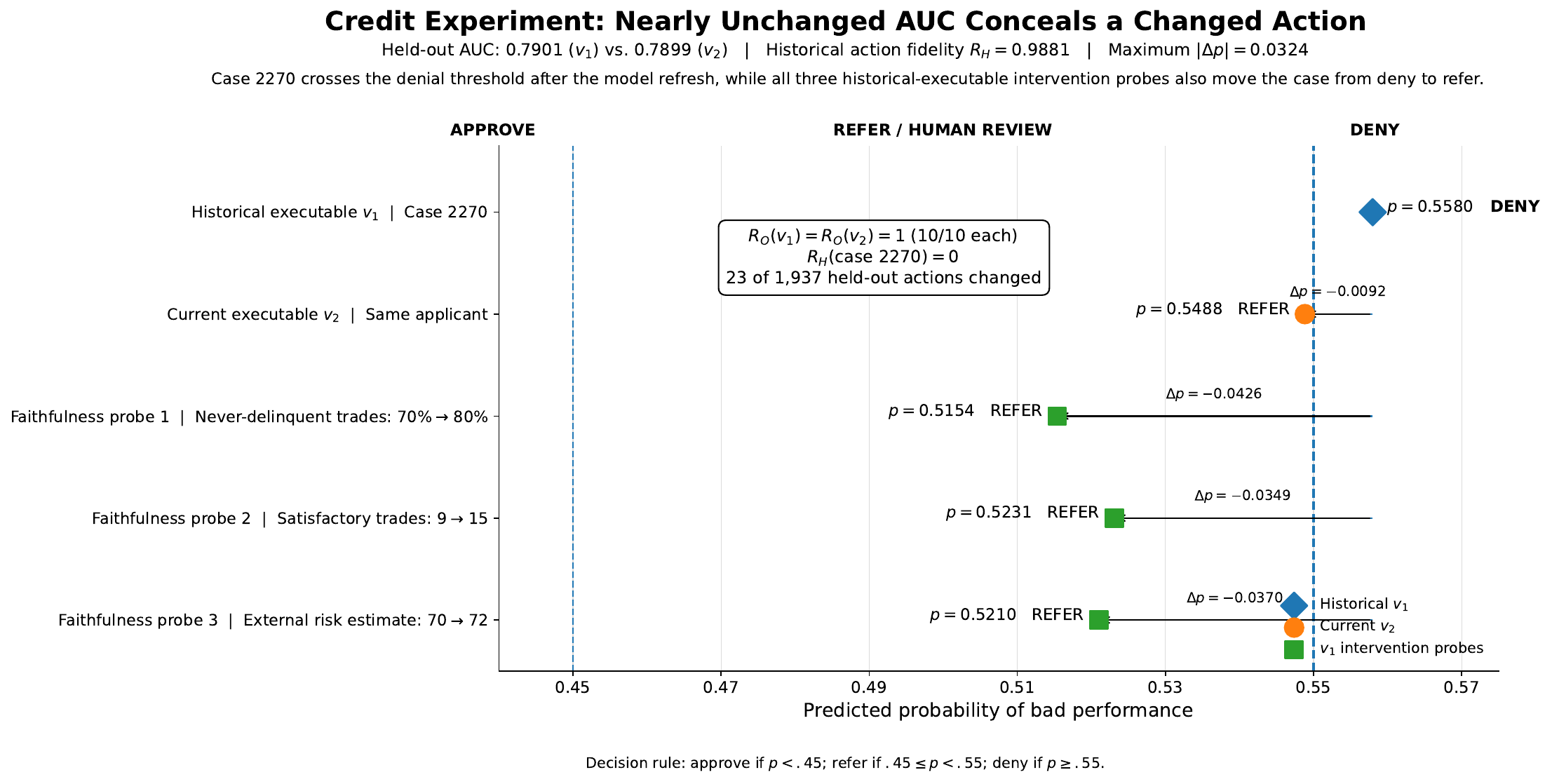}
\caption{\textit{Decision-level Verifiability Gap (case 2270).} Policy thresholds: approve $<0.45$, refer $[0.45, 0.55)$, deny $\ge 0.55$. Historical $v_1$ denies ($p=0.5580$) while updated $v_2$ refers ($p=0.5488$); a $0.0092$ shift alters the action despite perfect current stability ($R_O=1$ vs.\ $R_H=0$). Counterfactual probes on preserved $v_1$ confirm that stated reasons are causally testable rather than merely plausible.}
\label{fig:credit_verifiability_gap}
\end{figure}

\subsection{Study 4: The Same Mechanism on Real Credit Applications}
\label{subsec:heloc_orchestration}
Studies~1 and~2 rely on 32 constructed cases with known ground truth, leaving open whether the observed mechanisms are artifacts of synthetic inputs. Study~4 eliminates this concern by replicating the orchestration sweep on real credit applications.

\paragraph{Design.} We drew a stratified random sample of 32 applications from Study~3's held-out set across historical actions (12 denials, 10 referrals, 10 approvals) under a fixed seed. Crucially, none belong to the 23 boundary-shifting cases, precluding decision-boundary bias. Each application retains its five recorded credit attributes and uses the instructions and three-action vocabulary from Study~2.

The protocol mirrors Study~2 using locally served models to ensure strict decoding control. Configurations span 1 to 50 agents across \texttt{llama3.2:3b} and \texttt{llama3.1:8b} (five repetitions per case, varied seeds), tracking outcome reproducibility ($R_O$), exact trace reproducibility ($R_T$), and verdict differentiation ($D_V$). We additionally test pinned temperature-zero conditions at 1 and 10 agents. (Frontier commercial models are excluded here because their endpoints disallow necessary decoding controls; see Study~2).

\paragraph{The mechanism replicates, and the constructed cases understated
it.} Table~\ref{tab:heloc_sweep} reports the sweep. Wherever temperature
is not supplied, exact trace reproducibility is zero: every configuration, both model
scales, all 2,560 executions. The temperature-zero arm behaves differently and
is reported below. Outcome reproducibility falls and then recovers, as it does on the
constructed cases, but it falls further. On the smaller model it drops from
$0.8625$ at one call to $0.6188$ at ten agents, against $0.912$ and $0.794$ for
the constructed cases. The larger model follows the same path, turning later.

\paragraph{The recovery is again a collapse.} Where outcome reproducibility
recovers, verdict differentiation falls. The smaller model recovers to $0.7625$
at thirty agents while $D_V$ falls from $0.739$ to $0.439$, and at fifty agents
it denies 30 of the 32 applications. The larger model recovers to $0.8313$ at
fifty agents and refers 30 of 32. Both configurations report a respectable
outcome metric while answering almost every application the same way.

The two scales converge on the same degeneracy from different directions. At
fifty agents each reaches $D_V=0.213$ with 30 of 32 cases collapsed onto one
action. The action differs: the smaller model denies, the larger refers. Which
default a pipeline collapses onto appears to depend on the model. That it
collapses appears to depend on the architecture.

\paragraph{A perfect score on a system that has stopped deciding.} The temperature-zero arm
gives the sharpest case. With temperature pinned and a single call, the system
reproduces its own action in every repetition and its complete record in all 32
cases: $R_O=1.000$ and exact trace reproducibility $32/32$. A governance function
monitoring reproducibility would see a perfect score. The same configuration
approves all 32 applications, including the 12 the historical executable denies,
so $D_V=0.000$. Nothing about this failure is stochastic; the records are
byte-identical. The metric is at its maximum precisely where the system has
stopped deciding. Such a system is easy to verify (its rule is ``approve
everything''); what the result shows is that a reproducibility score cannot
certify fitness for delegated authority. $D_V$ is a diagnostic against
degenerate authorization, not a measure of $V_{dq}$, and the arm's
three-billion-parameter model is beside the point: the metric, not the model,
is what fails. Corollary~\ref{cor:gapbits} specifies the evidentiary
requirement for verdict-only retention; the sweep does not estimate a
numerical decision-level Verifiability Gap.

\begin{table}[!htpb]
\centering
\caption{Study 4. Orchestration sweep on 32 real held-out credit applications,
five repetitions per configuration. Exact trace reproducibility is zero in
every arm where temperature is not supplied.
\emph{Modal verdicts} counts how the 32 cases distribute across actions once
each case is summarized by its most frequent verdict.}
\label{tab:heloc_sweep}
\footnotesize
\renewcommand{\arraystretch}{1.02}
\begin{tabular}{@{}rrrrl@{}}
\toprule
\textbf{Agents} & $R_O$ & \textbf{Exact trace} & $D_V$ & \textbf{Modal verdicts} \\
 & & \textbf{reproducibility} & & \\
\midrule
\multicolumn{5}{@{}l}{\emph{\texttt{llama3.2:3b}, temperature not supplied}}\\
1  & 0.863 & 0/32 & 0.000 & approve 32 \\
2  & 0.663 & 0/32 & 0.739 & refer 22, deny 7, approve 3 \\
5  & 0.656 & 0/32 & 0.629 & deny 17, refer 15 \\
10 & 0.619 & 0/32 & 0.615 & deny 19, refer 13 \\
20 & 0.725 & 0/32 & 0.343 & deny 28, refer 4 \\
30 & 0.763 & 0/32 & 0.439 & deny 26, refer 6 \\
40 & 0.763 & 0/32 & 0.343 & deny 28, refer 4 \\
50 & 0.744 & 0/32 & 0.213 & deny 30, refer 2 \\
\midrule
\multicolumn{5}{@{}l}{\emph{\texttt{llama3.1:8b}, temperature not supplied}}\\
1  & 0.888 & 0/32 & 0.127 & approve 31, refer 1 \\
10 & 0.813 & 0/32 & 0.127 & refer 31, approve 1 \\
30 & 0.694 & 0/32 & 0.439 & refer 26, deny 6 \\
50 & 0.831 & 0/32 & 0.213 & refer 30, deny 2 \\
\midrule
\multicolumn{5}{@{}l}{\emph{\texttt{llama3.2:3b}, temperature pinned to zero}}\\
1  & \textbf{1.000} & \textbf{32/32} & \textbf{0.000} & approve 32 \\
10 & 0.981 & 21/32 & 0.804 & deny 15, refer 15, approve 2 \\
\bottomrule
\end{tabular}
\end{table}

\begin{figure}[!htbp]
\centering
\includegraphics[width=0.86\linewidth]{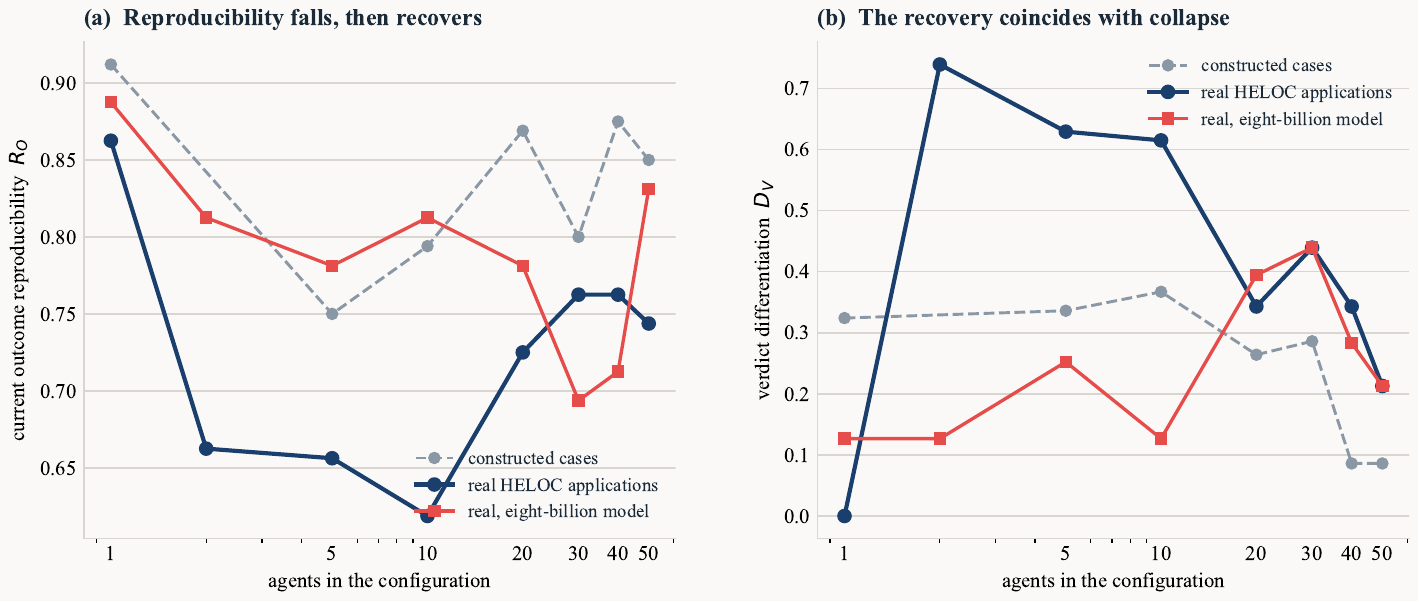}
\caption{Study 4 beside Study 2. \textbf{(a)} Outcome reproducibility falls and
then recovers on real applications as it does on constructed cases, and it falls
further. \textbf{(b)} Verdict differentiation starts wherever the model starts
and ends where the architecture takes it: at fifty agents both real-data arms
reach $D_V=0.213$ with 30 of 32 cases on a single action.}
\label{fig:study4}
\end{figure}
\section{Discussion}
\label{sec:discussion}

Our four empirical studies demonstrate evidentiary breakdown across three distinct loci, external provider infrastructure, multi-agent orchestration architectures, and temporal version boundaries, and replicate the orchestration mechanism on real credit applications. This cross-environment convergence indicates that the Verifiability Gap is not an isolated technical anomaly but a structural governance exposure. We synthesize these findings by detailing the theoretical contributions of evidence-contingent delegation, triangulating the empirical mechanisms of audit failure, outlining actionable governance mandates, and presenting a falsifiable research agenda.

\subsection{Theoretical Contributions}

\paragraph{Verifiability as an Authority--Evidence Relation.}
Prior information systems research examines algorithmic opacity primarily through model-centric lenses such as inscrutability, blackboxing, and post hoc explanation \citep{berente2021managing,kronblad2024blackboxing,asatiani2021envelopment}. We advance this tradition by establishing that in agentic workflows, explainability alone is insufficient: defensible governance requires material process reconstruction. Because an agentic system distributes decision authority across models, tools, prompts, and orchestration layers, verifiability is not an invariant system attribute. Rather, it is a dynamic relation ($G_{dq} = [\rho_{\sigma}(A_d) - V_{dq}]_{+}$) between the authority exercised and the material evidence retained to satisfy a specific audit standard $\sigma$.

\paragraph{Reproducibility as a Multidimensional Governance Profile.}
We reject the assumption that reproducibility can be captured as a scalar metric of output consistency. Conventional reporting of current outcome reproducibility ($R_O$) generates false assurances. By formalizing reproducibility as a multidimensional profile $\mathcal{R}(t;\sigma) = (R_O, R_H, R_P^H, R_P^C, R_T; D_V)$, our theory reveals that: (1) high current stability can mask complete historical replay failure ($R_O = 1$ while $R_H = 0$); (2) identical terminal actions frequently decouple from divergent underlying traces ($R_P^C = 0$); and (3) output consistency can artificially surge precisely when a system loses its capacity to distinguish cases ($D_V \to 0$). Intelligible explanations possess evidentiary validity only when causally anchored to the specific executable that rendered the original judgment.

\paragraph{Evidence-contingent delegation (ECD).}
Extending foundational theories of algorithmic delegation and control \citep{baird2021delegation,kellogg2020algorithms}, our framework establishes that institutional authority ceilings are bounded by retained verification capacity rather than raw technical capability (P1). Under ECD, delegated autonomy cannot be treated as a permanent operational right granted at onboarding. Instead, it operates as a revocable governance relationship contingent on continuous evidence preservation. Provider updates, orchestration reconfigurations, and execution-control deprecations are therefore not routine technical maintenance; they are material governance shocks that mandate formal reauthorization.

\subsection{Empirical Synthesis: Where Governance Fails}
\label{subsec:evidence_synthesis}

Across our four studies, traditional validation metrics overstated retained verification capacity along four operational dimensions. \textit{First, Provider Discontinuity (Study 1):} Even when an institution preserves internal code and inputs, upstream vendor releases shifted baseline-modal actions ($\widetilde{R}_H = 0.906$) and withdrew decoding controls (e.g., zero-temperature rejection), removing replay capacity from outside firm boundaries. \textit{Second, Architecture as Policy (Study 2):} Interposing agent layers altered terminal risk postures without policy changes; across sweeps from 1 to 50 agents, exact trace reproducibility ($R_T$) collapsed to zero while outcome agreement exhibited an artificial rebound driven by categorical compression ($D_V \to 0.086$). \textit{Third, Historical Replay Failure (Study 3):} In a fully deterministic credit model with a stable aggregate AUC ($\Delta\mathrm{AUC} = 0.0002$), 23 held-out actions flipped across version boundaries ($R_H = 0.9881$). Marginal applicants exhibited absolute current stability ($R_O = 1$) alongside total historical replay failure ($R_H = 0$), rendering updated-model rationales evidentially nonresponsive to past denials. \textit{Fourth, Reproducibility Without Differentiation (Study 4):} Evaluating 32 real FICO HELOC applications confirmed that orchestration-driven compression is not a synthetic artifact: unconstrained models collapsed 30 of 32 applications onto uniform defaults ($D_V = 0.213$), while pinned zero-temperature execution yielded a ``perfect'' audit score ($R_O = 1.000, R_T = 32/32$) solely because the system stopped differentiating and approved every case ($D_V = 0.000$).

\paragraph{Ruling Out Competing Explanations.}
These four studies eliminate common technical rationalizations. Study~3 is deterministic, inspectable, and single-step, proving that replay failure requires neither generative stochasticity nor multi-agent orchestration. Study~2 fixes the underlying model and provider, proving that multi-agent architecture alone acts as a latent policy layer. Study~1 fixes the prompts, policy, and test cases, proving that diligence within the firm cannot prevent provider-driven audit failure. Finally, Study~4 replicates these dynamics on real credit applications, confirming that categorical collapse arises under the serial architectures tested, on real as well as constructed cases; isolating serial compression from prompting, context length, and capacity requires the architectural comparison noted in the limitations.

\subsection{Implications for Agentic AI Governance}

Our findings yield three foundational mandates for financial institutions and regulators. \textit{Transition to event-driven reauthorization:} Periodic validation proportionate to model risk \citep{srguidance2011,srguidance2026} is insufficient for agentic architectures; reauthorization must be triggered dynamically whenever an upstream foundation model updates, an orchestration graph changes, or an external API modifies execution controls (P4). \textit{Mandate executable evidence bundles:} Passive compliance artifacts like model cards and static prompt logs \citep{gebru2018datasheets,mitchell2019model} fail to substantiate agentic actions (P5); institutions must preserve complete \emph{executable evidence bundles}, including containerized local model weights, exact tool payloads, temperature and seed states, intermediate handoff transcripts, and operative policy thresholds. Where a hosted provider will not release the executable or decoding controls (Study~1), the remaining lever is to delegate less or serve the model locally. \textit{Mitigate network weak-link and common-mode risks:} End-to-end auditability across multi-firm networks is strictly bounded by the weakest unrecorded handoff (P6). The concentration documented in HMDA filings (Appendix~\ref{app:hmda}) motivates evaluating whether shared system dependence is accompanied by inadequate historical replay controls. Under that additional condition, material upstream changes could create correlated verification losses (P7); these dynamics are not tested here.

\subsection{Limitations and Future Research}
\label{subsec:limitations}

Several boundary conditions contextualize our findings. First, our multi-agent sweeps utilize 32 constructed cases (Studies 1 and 2) and 32 held-out empirical credit cases (Study 4) to isolate causal mechanisms under controlled repetition regimes rather than estimate sector-wide default frequencies; wider industry studies are needed to measure prevalence across broader portfolio distributions. Constructed cases are required by Theorem~\ref{thm:identify}: degeneracy cannot be detected from within-case statistics, so the design needs families whose correct actions differ. The orchestration graph forces a single-label terminal decision, which mirrors deployed pipelines \citep{xiao2024tradingagents} but also eases categorical compression; a planner-style comparison is left to future work. Local sweeps use three- and eight-billion-parameter models because the frontier endpoint rejects the decoding controls the design requires. Second, a fixed hosted model ID does not preserve the full serving environment or guarantee historical byte-identical replay. Study~1 distinguishes cross-version action drift from restricted decoding control; it does not establish silent replacement of a fixed model snapshot. Third, $E_{dq}$ is not measured; the gap is instantiated through reproducibility only. Fourth, exact trace reproducibility ($R_T$) enforces a strict textual standard; future work should refine semantic distance metrics ($\equiv_\sigma$) to delineate material legal discrepancy from cosmetic variation.

Finally, while Studies 1--4 test the system-internal mechanics underlying P4, P5, and P6, our organizational and network propositions define a concrete agenda for empirical field falsification: (1) \textit{testing authority ceilings (P1)} by examining whether model-risk committee approvals track technical capability once retained verification capacity is held constant; (2) \textit{testing control substitution (P2--P3)} by measuring whether human sign-off rates remain constant despite increasing unfaithfulness in agentic rationales; and (3) \textit{testing common-mode exposure (P7)} by investigating whether supervisory audit failures correlate across institutions following major provider updates. Beyond this agenda, longitudinal studies should measure API auditability half-lives and examine human underwriter reliance on unfaithful agentic narratives.

\section{Conclusion}
\label{sec:conclusion}

As FinTech transitions from predictive algorithms to autonomous multi-agent workflows, financial institutions operate under the assumption that delegated decision authority remains explainable and reproducible. This study challenges that premise by conceptualizing and formalizing the \emph{Verifiability Gap}: the structural deficit between the verification demanded to defend an exercised authority and the material evidence retained to reconstruct that action for an auditor. By shifting the analytical lens of Information Systems governance from static \emph{model transparency} to dynamic \emph{system verifiability}, we demonstrate that fairness and verifiability are orthogonal. A system can satisfy portfolio-level fairness audits while leaving individual, consequential credit denials or AML escalations unsubstantiable.

Our empirical findings establish that technical capability does not equate to auditability. A deterministic credit model reproduced its own current actions yet failed to recover historical actions at policy cutoffs; similarly, frontier LLMs achieved high outcome consistency while never repeating a single complete execution trace. The sharpest finding for practice is reproducibility without differentiation: a system achieved a perfect audit score ($R_O = 1.000, R_T = 32/32$) precisely when it stopped deciding, approving all 32 credit applications under zero temperature, including the 12 the historical model had denied. Output-agreement metrics can therefore validate and certify an agentic pipeline at the exact moment it loses the capacity to distinguish cases. Across finance, healthcare, and public administration, reproducibility must be evaluated as a multidimensional profile that measures whether an agent decided, rather than whether it returned the same default twice.

\paragraph{The ECD Principle.}
High model capability and surface outcome stability offer no guarantee of agentic auditability. We therefore advance the theoretical principle of \emph{evidence-contingent delegation}: delegated AI authority remains defensible only so long as retained verification capacity stays commensurate with the authority exercised ($V_{dq} \geq \rho_{\sigma}(A_d)$). Delegated autonomy is never settled once at initial onboarding; it operates as an ongoing, revocable governance relationship that must be re-earned whenever the underlying model, orchestration topology, or execution-control surface changes. When historical evidence cannot be retained, the defensible institutional response is not performative procedural review, but the reduction of delegated authority. By anchoring autonomous delegation in material verifiability rather than static capability benchmarks, organizations can harness the scale of agentic AI while safeguarding legal accountability and regulatory defensibility.

\section*{Declarations}


\noindent\textbf{Declaration of competing interest.} The author declares no
competing interests.

\noindent\textbf{Data availability.} All code, decision cases, and execution logs reproducing every reported result and hypothesis test are available in the anonymized repository at \url{https://anonymous.4open.science/status/AgenticAI-AEAF}. Complete run outputs are preserved, enabling full offline verification without external API keys. The HMDA analysis in Appendix~\ref{app:hmda} starts from about 12~GB of public FFIEC loan-level files; the package includes the extraction script and the derived 24.4-million-row extract from which every reported figure is recomputed.

\clearpage
\begingroup\footnotesize
\setlength{\bibsep}{0pt}\setlength{\itemsep}{0pt}

\endgroup

\clearpage
\begin{center}
{\sffamily\LARGE\bfseries Appendix}\par\vskip 4pt
{\sffamily\small Supplementary configuration, proofs, calibration, and provider-concentration data}
\end{center}
\vskip 1.2em
\appendix
\renewcommand{\thesection}{\Alph{section}}
\setcounter{section}{0}
\setcounter{table}{0}\renewcommand{\thetable}{A\arabic{table}}
\setcounter{figure}{0}\renewcommand{\thefigure}{A\arabic{figure}}

\noindent
This appendix supports the main text. Appendix~\ref{app:agents} details the multi-agent experimental architecture for Studies~2 and~4. Appendix~\ref{app:proofs} provides formal proofs, simplified derivations, and empirical calibrations. Appendix~\ref{app:hmda} documents provider concentration in U.S. mortgage underwriting.

\section{The Multi-Agent Experiment}
\label{app:agents}


\textit{Decoding Conditions.} Auditing must evaluate factors outside institutional control. A fixed random seed enforces computational determinism under a fixed executable, but deployment seeds are rarely exposed or preserved. Pinning temperature to zero provides maximal lexical constraint, yet commercial providers frequently deprecate it (Study~1). When omitted, models revert to provider defaults. In Conditions~B and~C, all chain calls share a seed, understating independent divergence; Condition~D eliminates all controls, mirroring live audit conditions (Table~\ref{tab:decoding}).

\begin{table}[!ht]
\centering
\caption{Decoding conditions across experimental arms.}
\label{tab:decoding}
\footnotesize
\begin{tabular}{@{}p{1.1cm}p{2.6cm}p{2.3cm}p{7.4cm}@{}}
\toprule
& \textbf{Temperature} & \textbf{Seed} & \textbf{What it measures} \\
\midrule
A & pinned to 0 & fixed & Implementation determinism. \\
B & pinned to 0 & varied & Seed sensitivity under pinned temperature. \\
C & not sent & varied & Sensitivity to unconstrained temperature alone. \\
D & not sent & not sent & Live audit setting: neither control is accessible. \\
\bottomrule
\end{tabular}
\end{table}

\textbf{Multi-Agent Pipeline Architecture.} The pipeline enforces strict serial governance: Layer~1 alone inspects the case, intermediate specialists are strictly report-only, all intermediate handoffs are recorded, and the terminal Decider alone renders the final action. Scaled topologies expand specialist call volume under these exact constraints. Preserving every handoff allows exact localization of where divergence originates and whether downstream stages amplify or compress it. Table~\ref{tab:agentroles} outlines canonical roles; replication manifests specify exact call graphs.

\begin{table}[!ht]
\centering
\caption{Canonical role functions in the controlled multi-agent configuration.}
\label{tab:agentroles}
\footnotesize
\begin{tabular}{@{}p{1.7cm}p{1.0cm}p{2.7cm}p{7.7cm}@{}}
\toprule
\textbf{Agent} & \textbf{Layer} & \textbf{Reads} & \textbf{Instruction} \\
\midrule
retrieval  & 1 & the case & Extract decision-relevant facts and numbers. Do not judge. \\
market     & 1 & the case & Summarize market and environment signals. Do not judge. \\
policy     & 1 & the case & State applicable policy thresholds and rules. Do not judge. \\
analyst    & 2 & retrieval, policy & Quantitatively evaluate against thresholds. Do not act. \\
risk       & 2 & retrieval, market & Detail material risks and overall severity. Do not act. \\
compliance & 2 & retrieval, policy & Flag regulatory or fair-treatment concerns. Do not act. \\
aml        & 2 & retrieval & Flag financial-crime or counterparty risks. Do not act. \\
quant      & 3 & analyst, risk & Score options on expected benefit and downside. Do not act. \\
customer   & 3 & compliance, aml & Detail customer or counterparty impact. Do not act. \\
critic     & 4 & quant, customer & Challenge emerging consensus with counterarguments. Do not act. \\
\midrule
decider    & 5 & quant, customer, critic & Choose one allowed action and provide a one-sentence rationale. \\
\bottomrule
\end{tabular}
\end{table}
\textbf{Execution, Scope, and Replication Checklist}

\textit{Execution and Scope.}
All runs executed locally on \texttt{llama3.2:3b} using 32 fixed financial cases (eight each across credit, trade, AML, and rebalancing), eliminating provider-side confounders. The pipeline fixes communication graphs, freezes tool payloads, and executes a single acyclic pass without live retrieval, persistent memory, or autonomous re-prompting. Findings therefore represent a conservative lower bound on architecture-induced divergence.

\textit{Reproduction Checklist.}
Before replication, verify: (1) case hash matches \texttt{696afd10}; (2) model digests and prompts match repository manifests; (3) parsers accept only authorized family-specific action vocabularies; (4) temperature and seed settings conform to the specified experimental arm; and (5) tool payloads match frozen hashes. Hosted provider runs cannot be re-executed identically due to provider endpoint updates; their historical request logs, timestamps, and model responses are archived for offline audit.

\section{Proofs and Formal Calibration}
\label{app:proofs}

This appendix proves the results from Sections~\ref{subsec:reversibility} and~\ref{subsec:operationalization} and calibrates them against empirical data.

\subsection{Notation}
\label{app:notation}

Calls are indexed $i=1,\dots,m$ in topological order; $S_k$ denotes layer-$k$ outputs, $T$ the execution trace, and $Y=\Gamma(T)$ the action recovered by a fixed parser $\Gamma$ from the terminal Decider. Pairwise trace and outcome agreement across repetitions are $A_T=\Pr[T^{(1)}=T^{(2)}]$ and $A_O=\Pr[Y^{(1)}=Y^{(2)}]$. While $A_O$ charges every disagreeing pair, the modal share estimator $\widehat R_O = |\mathcal C|^{-1}\sum_c\exp(-H_\infty(\hat \nu_c))$ charges only non-modal mass, giving $A_O \le \widehat R_O$. Call conditional self-agreement is $p_i=\Pr[Z_i^{(1)}=Z_i^{(2)}\mid Z_j^{(1)}=Z_j^{(2)}\ \forall j<i]$, and terminal compression probability is $\kappa=\Pr[Y^{(1)}=Y^{(2)}\mid T^{(1)}\neq T^{(2)}]$.

\begin{lemma}
\label{lem:book}
$(i)$~$A_T=\prod_{i=1}^{m}p_i$, so $A_T$ is non-increasing in $m$ and $A_T\le\bar p^{\,m}$ whenever $p_i\le\bar p$.
$(ii)$~$A_O=A_T+\kappa(1-A_T)\ge A_T$; in particular $A_O=\kappa$ when $A_T=0$.
$(iii)$~$\widehat R_O=|\mathcal C|^{-1}\sum_c\exp(-H_\infty(\hat \nu_c))$, where $\nu_c$ is the verdict distribution across repetitions of case $c$ and $H_\infty(\nu)=-\log\max_\ell \nu_\ell$.
\end{lemma}

\begin{proof}
$(i)$~Let $A_i=\{Z_i^{(1)}=Z_i^{(2)}\}$. Then exact trace reproducibility is $\bigcap_{i\le m} A_i$, and the chain rule gives $\Pr[\bigcap_i A_i]=\prod_i\Pr[A_i\mid\bigcap_{j<i}A_j]=\prod_i p_i$ by the definition of $p_i$.
$(ii)$~$Y$ is a deterministic function of $T$, so $T^{(1)}=T^{(2)} \implies Y^{(1)}=Y^{(2)}$. Conditioning on $\{T^{(1)}=T^{(2)}\}$ and its complement yields $A_O = 1\cdot A_T + \kappa(1-A_T)$.
$(iii)$~Follows directly from $\max_\ell \nu_\ell = \exp(-H_\infty(\nu))$ averaged over $|\mathcal C|$.
\end{proof}

\subsection{Proofs of the Reversibility Results}
\label{app:rev_statements}

\begin{remark}
Reversibility denotes whether the historical decision process can be reconstructed from retained records, not thermodynamic detailed balance.
\end{remark}

\begin{corollary}
\label{cor:norepair}
No downstream agent, aggregation rule, critic, or terminal decider can raise $I(C;Y)$ above $I(C;S_1)$. Case information not captured by the first layer is unrecoverable by any later stage, however capable.
\end{corollary}

\begin{theorem}[Material reversibility and case-information sufficiency]
\label{thm:revdpi}
Under the Markov relation $C\to S_{k-1}\to S_k$, material reversibility without retained records requires $H([S_{k-1}]_\sigma\mid S_k)=0$. Separately, equality in the data-processing inequality, $I(C;S_k)=I(C;S_{k-1})$, holds if and only if $I(C;S_{k-1}\mid S_k)=0$. These are distinct criteria: neither implies the other without additional assumptions linking $\sigma$-material distinctions to information about $C$.
\end{theorem}

\begin{theorem}[Specialization and reversibility are incompatible]
\label{thm:incompat}
If stage $k$ performs summarization in the ordinary sense, deterministic and not injective modulo $\equiv_\sigma$, then the stage is not $\sigma$-reversible with a trivial record, and there exist materially distinct upstream states the verifier cannot separate from the downstream record. When the collapsed distinction carries information about $C$, the loss is strict for that case distribution, $I(C;S_k)<I(C;S_{k-1})$, testable as $I(C;S_{k-1}\mid S_k)>0$. This concerns the case distribution, not the channel coefficient: a projection discarding a variable independent of $C$ compresses without lowering $\eta_k$, so geometric decay in Theorem~\ref{thm:contraction} rests on the assumption $\eta_k\le\bar\eta<1$, not on summarization alone. Conversely, a stage that is $\sigma$-reversible with a trivial record performs no material compression; it merely relabels, and it need not preserve upstream detail that $\sigma$ declares immaterial.
\end{theorem}

\begin{corollary}[Retention must grow with depth]
\label{cor:retentiondepth}
Each materially compressing stage contributes a strictly positive term to the retention bound $H(B)\ge H([T]_\sigma\mid Y)$. The retention required to hold $G_{dq}=0$ is therefore non-decreasing in the number of causally material stages and strictly increasing in the number that compress.
\end{corollary}

\begin{proof}[Proof of Theorem~\ref{thm:contraction}]
Because no layer beyond the first reads $C$, the operative sequence forms a Markov chain $C \to S_1 \to \cdots \to S_N \to Y$. Iterating the strong data-processing inequality, where $\eta_k < 1$ whenever the channel maps distinct inputs to overlapping distributions, yields $I(C;Y)\le\bigl(\prod_{k\ge2}\eta_k\bigr) I(C;S_1)$. When $\eta_k\le\bar\eta<1$, case information decays geometrically. Corollary~\ref{cor:norepair} is the immediate case for $k>1$.
\end{proof}

\begin{proof}[Proof of Theorem~\ref{thm:retention}]
Reversibility requires $[T]_\sigma$ to be a function of $(Y,B)$ almost surely, meaning $H([T]_\sigma\mid Y,B)=0$. Expanding conditional entropy:
\[
H([T]_\sigma\mid Y) \le H([T]_\sigma, B\mid Y) = H(B\mid Y) + H([T]_\sigma\mid Y,B) = H(B\mid Y) \le H(B).
\]
Achievability follows by encoding the conditional pre-image within each stage. Corollary~\ref{cor:gapbits} substitutes these quantities into $G_{dq}$, and Corollary~\ref{cor:retentiondepth} follows from Theorem~\ref{thm:incompat}.
\end{proof}

\begin{proof}[Proof of Theorem~\ref{thm:revdpi}]
The map $S_k\mapsto[S_{k-1}]_\sigma$ of Definition~\ref{def:rev} is well defined exactly when $H([S_{k-1}]_\sigma\mid S_k)=0$. Under the Markov relation, $I(C;S_{k-1})-I(C;S_k)=I(C;S_{k-1}\mid S_k)$, so equality in the data-processing inequality holds if and only if $I(C;S_{k-1}\mid S_k)=0$. The first criterion concerns $\sigma$-classes, the second information about $C$; a map can be injective modulo $\equiv_\sigma$ while $C$ depends on within-class detail, and conversely, so neither implies the other without a linking condition.
\end{proof}

\begin{proof}[Proof of Theorem~\ref{thm:incompat}]
Non-injectivity modulo $\equiv_\sigma$ implies there exist $s\not\equiv_\sigma s'$ mapping to identical outputs. When $C$ depends on this distinction, $I(C;S_k)<I(C;S_{k-1})$ is strict for that distribution. The converse, that a $\sigma$-reversible stage with a trivial record performs no material compression, is the first criterion of Theorem~\ref{thm:revdpi}.
\end{proof}

\subsection{Proofs of the Metric Results}
\label{app:metric_statements}

\begin{theorem}[Range mismatch]
\label{thm:range}
Let $L=|\mathcal L|$ be the number of allowed actions, $m$ the model call count, and $\bar p<1$ the maximum single-call repeatability. Then outcome agreement is bounded below, $A_O\ge 1/L$ and $\widehat R_O\ge 1/L$, whereas $A_T\le\bar p^{\,m}\to0$ and $I(C;Y)\to0$ geometrically in depth. As depth grows, the attainable region of $(A_T,I(C;Y),A_O)$ approaches $\{0\}\times\{0\}\times[1/L,1]$.
\end{theorem}

\begin{corollary}[No outcome threshold is a sufficient audit criterion]
\label{cor:nothreshold}
For any authorization rule requiring an outcome statistic above threshold $\alpha$, there exist degenerate configurations satisfying it with $A_T=0$ and $I(C;Y)=0$.
\end{corollary}

\begin{theorem}[Degeneracy is not identifiable within cases]
\label{thm:identify}
Let $\mathcal P_\star$ return the case-appropriate action $\ell^\star(c)$, and let $\mathcal P_0$ return a constant $\ell_0$. Every within-case statistic invariant under permutations of action labels evaluates identically under these two point-mass distributions. This includes modal share, pairwise agreement, and entropy, but not label-sensitive statistics or comparison with a known correct action.
\end{theorem}

\begin{corollary}
\label{cor:identification}
Distinguishing $\mathcal P_\star$ from $\mathcal P_0$ requires at least two cases with differing correct actions and a cross-case dispersion statistic. Decision families with differentiated actions, paired with $D_V$, constitute the minimal design identifying degeneracy.
\end{corollary}

\begin{remark}
\label{rem:two-dists}
Outcome agreement ($A_O, \widehat R_O$) reflects within-case distributions, while $D_V$ and $I(C;Y)$ evaluate cross-case modal distributions. The former are blind to categorical degeneracy by construction.
\end{remark}

\begin{proof}[Proof of Theorem~\ref{thm:range}]
For any distribution over $L$ actions, $\max_\ell \nu_\ell\ge1/L$, so $\widehat R_O\ge1/L$ by Lemma~\ref{lem:book}$(iii)$; at the population level $\sum_\ell \nu_\ell^2\ge1/L$ by Cauchy--Schwarz. Limits follow from Lemma~\ref{lem:book}$(i)$ and Theorem~\ref{thm:contraction}. The constant policy $Y\equiv\ell_0$ achieves $A_O=\widehat R_O=1$ with $I(C;Y)=0$, proving Corollary~\ref{cor:nothreshold}.
\end{proof}

\begin{proof}[Proof of Theorem~\ref{thm:identify}]
Under $\mathcal P_\star$, $\nu_c=\delta_{\ell^\star(c)}$; under $\mathcal P_0$, $\nu_c=\delta_{\ell_0}$. Since both are point masses, empirical samples over $K$ runs yield identical single-point distributions, rendering any within-case functional insensitive to case distinctions.
\end{proof}

\subsection{Non-Monotonicity of Outcome Reproducibility}

\begin{theorem}
\label{thm:nonmono}
With $\Delta f(n)=f(n{+}1)-f(n)$,
\begin{equation}
\Delta A_O(n)=
\underbrace{(1-\kappa_{n+1})\,\Delta A_T(n)}_{\le\,0}
+\underbrace{(1-A_T(n))\,\Delta\kappa(n)}_{\text{sign of }\Delta\kappa}.
\label{eq:nonmono}
\end{equation}
Whenever terminal mapping coarsening causes the second term to exceed the first, $A_O$ rises while $A_T$ falls.
\end{theorem}

\begin{proof}
Differencing $A_O(n)=A_T(n)+\kappa_n(1-A_T(n))$ from Lemma~\ref{lem:book}$(ii)$ and collecting terms yields \eqref{eq:nonmono}.
\end{proof}

Because $A_T=0$ holds everywhere under released temperatures, $\Delta A_T(n)=0$ and $\Delta A_O(n)=\Delta\kappa(n)$ holds exactly. Thus, observed recoveries in outcome reproducibility reflect terminal compression alone.

\subsection{Empirical Calibration and Model Scale}
\label{app:calibration_and_scale}

Calibrating Table~\ref{tab:calib} confirms four theoretical signatures:
(i)~\textit{Decoupling:} at $A_T=0$, $A_O = \kappa$ by construction;
(ii)~\textit{Stability vs.\ case information:} peak agreement ($A_O = 0.7875$, $n=40$) matches minimal case information ($I(C;Y) = 0.0853$ bits), confirming Corollary~\ref{cor:nothreshold};
(iii)~\textit{Retention bound:} the bound concerns $H([T]_\sigma\mid Y)$, not $I(C;Y)$. The sweep does not estimate this conditional material-trace entropy, so it does not establish a 1.5--1.8-bit decision-level gap. Every handoff was retained in the experiment; low case information alone therefore does not establish positive residual uncertainty $H([T]_\sigma\mid Y,B)$; and
(iv)~\textit{Error compounding:} on \texttt{llama3.1:8b} ($T=0$), single-call self-agreement $\hat p \approx 0.9688$ predicts an 11-call $A_T=0.70$ versus the observed $A_T=0.45$ over context length.

Scale cannot resolve these deficits. Parameter count governs only $p_i$ and $\eta_k$ (neither monotonic in scale); when $p_i < 1$, process reproducibility collapses geometrically ($R_T = \prod p_i$). Thus, larger models in deeper topologies do not inherently improve auditability.

\section{Provider Concentration in U.S. Mortgage Underwriting}
\label{app:hmda}

To evaluate common-mode network exposure (P7), we processed the public HMDA loan-level files for 2019--2023 across six states (CA, TX, NY, FL, IL, PA), about 12~GB of raw records, retaining every application with a credit decision ($24{,}439{,}380$ decisions, $4{,}487$ lenders, $20.1\%$ denial rate); these reveal high automated underwriting system (AUS) concentration (Table~\ref{tab:hmda_conc}).

\begin{table}[!ht]
\centering
\caption{AUS shares in U.S. mortgage decisions (2019--2023, six states). Shares reflect AUS-evaluated loans; HHI is calculated over system shares.}
\label{tab:hmda_conc}
\footnotesize
\renewcommand{\arraystretch}{1.02}
\begin{tabular}{@{}rrrrrrr@{}}
\toprule
\textbf{Year} & \textbf{AUS-decided} & \textbf{Desktop} & \textbf{Loan Product} & \textbf{TOTAL} & \textbf{Top two} & \textbf{HHI} \\
 & \textbf{applications} & \textbf{Underwriter} & \textbf{Advisor} & \textbf{Scorecard} & & \\
 & & \emph{(Fannie Mae)} & \emph{(Freddie Mac)} & \emph{(FHA)} & & \\
\midrule
2019 & 2{,}929{,}098 & 63.9\% & 15.9\% & 8.2\%  & 79.8\% & 4{,}535 \\
2020 & 4{,}818{,}652 & 64.8\% & 23.9\% & 5.2\%  & 88.7\% & 4{,}832 \\
2021 & 5{,}149{,}015 & 64.4\% & 25.6\% & 4.4\%  & 90.0\% & 4{,}834 \\
2022 & 2{,}085{,}854 & 61.8\% & 22.2\% & 8.6\%  & 84.0\% & 4{,}416 \\
2023 & 1{,}582{,}341 & 55.8\% & 21.9\% & 11.9\% & 77.7\% & 3{,}786 \\
\bottomrule
\end{tabular}
\end{table}

Two GSE platforms (Fannie Mae's Desktop Underwriter, Freddie Mac's Loan Product Advisor) evaluated $77.7\%$--$90.0\%$ of automated decisions annually ($\mathrm{HHI} \in [3{,}786, 4{,}834]$; 1,674 and 1,241 institutional users in 2021, respectively). Within-institution concentration is equally high: of 2,376 AUS lenders, $52.6\%$ routed $>90\%$ and $35.9\%$ routed $>99\%$ of decisions through a single platform. These data establish shared upstream reliance, $u_{ig}$, but do not measure historical replay control, $\theta_{ig}$, or identify model-version changes; they document a structural precondition for common-mode exposure, not positive uncontrolled exposure or observed synchronized audit failures (P7).
\end{document}